%% file: arxiv_submission.tex
\documentclass{article}

\usepackage{arxiv}

\usepackage[utf8]{inputenc} 
\usepackage[T1]{fontenc}    
\usepackage{hyperref}       
\usepackage{url}            
\usepackage{booktabs}       
\usepackage{amsfonts}       
\usepackage{nicefrac}       
\usepackage{microtype}      
\usepackage{graphicx}
\usepackage{authblk}
\usepackage[style=authoryear, natbib=true, backend=bibtex,url=false,doi=true]{biblatex}
\usepackage{doi}

\usepackage{wrapfig}
\usepackage{algorithm,algpseudocode}
\usepackage{multicol,multirow}

\usepackage[most]{tcolorbox}
\usepackage{moreverb}
\usepackage{amscd,dsfont}
\usepackage{subcaption}
\usepackage{amsmath}
\usepackage{amssymb}       
\usepackage{amsthm}
\usepackage{thmtools}
\usepackage{makecell}
\usepackage{mathtools}
\usepackage{commath}
\usepackage{tikz}
\usepackage{xcolor}
\usepackage{boldline}
\usepackage{adjustbox}
\usepackage{array}
\usepackage{pgfplots,pgffor}
\usepackage{pgfplotstable}
\usepackage{tikzscale}
\usepackage{pifont}
\usepackage{orcidlink}

\makeatletter
\DeclareCiteCommand{\citetalias}
  {\usebibmacro{prenote}}
  {\usebibmacro{citeindex}%
   \printtext[bibhyperref]{\@citealias{\thefield{entrykey}}}}
  {\multicitedelim}
  {\usebibmacro{postnote}}

\DeclareCiteCommand{\citepalias}[\mkbibparens]
  {\usebibmacro{prenote}}
  {\usebibmacro{citeindex}%
   \printtext[bibhyperref]{\@citealias{\thefield{entrykey}}}}
  {\multicitedelim}
  {\usebibmacro{postnote}}
  \makeatother
  \defcitealias{ellam2018}{Ellam}
  \defcitealias{gaskin2023}{Gaskin}

  \AtEveryBibitem{
    \clearfield{urlyear}
    \clearfield{urlmonth}
    }

\newtheorem{theorem}{Theorem}[section]
\newtheorem{proposition}{Proposition}[section]

\theoremstyle{definition}
\newtheorem{definition}{Definition}[section]

\newtheorem{corollary}{Corollary}[theorem]
\newtheorem{lemma}[theorem]{Lemma}

\theoremstyle{remark}
\newtheorem{remark}[theorem]{Remark}

\input{./my_macros.tex}

\title{Table Inference for Combinatorial Origin-Destination Choices in Agent-Based Population Synthesis}

\date{} 					

\author[*,1]{ { Ioannis Zachos} \orcidlink{0000-0002-7503-2117}}
\author[2,3]{ { Theodoros Damoulas} \orcidlink{0000-0002-7172-4829}} 
\author[1,3]{ { Mark Girolami} \orcidlink{0000-0003-3008-253X}}

\affil[1]{ { Department of Engineering, University of Cambridge, Cambridge, UK.} }
\affil[2]{ { Department of Computer Science and Department of Statistics, University of Warwick, Coventry, UK.} }
\affil[3]{ { The Alan Turing Institute, London, UK.} }

\hypersetup{
pdftitle={Table Inference for Combinatorial Origin-Destination Choices in Agent-Based Population Synthesis},
pdfsubject={stat.ME},
pdfauthor={Ioannis Zachos, Theodoros Damoulas, Mark Girolami},
pdfkeywords={Combinatorial Explosion, Origin-destination matrices, Markov Bases, Spatial Interaction Models, Population synthesis},
}

\begin{document}

\maketitle

\begin{abstract}
	A key challenge in agent-based mobility simulations is the synthesis of individual agent socioeconomic profiles. Such profiles include locations of agent activities, which dictate the quality of the simulated travel patterns. These locations are typically represented in origin-destination matrices that are sampled using coarse travel surveys. This is because fine-grained trip profiles are scarce and fragmented due to privacy and cost reasons. The discrepancy between data and sampling resolutions renders agent traits non-identifiable due to the combinatorial space of data-consistent individual attributes. This problem is pertinent to any agent-based inference setting where the latent state is discrete. Existing approaches have used continuous relaxations of the underlying location assignments and subsequent ad-hoc discretisation thereof. We propose a framework to efficiently navigate this space offering improved reconstruction and coverage as well as linear-time sampling of the ground truth origin-destination table. This allows us to avoid factorially growing rejection rates and poor summary statistic consistency inherent in discrete choice modelling. We achieve this by introducing joint sampling schemes for the continuous intensity and discrete table of agent trips, as well as Markov bases that can efficiently traverse this combinatorial space subject to summary statistic constraints. Our framework's benefits are demonstrated in multiple controlled experiments and a large-scale application to agent work trip reconstruction in Cambridge, UK.
\end{abstract}

\keywords{Combinatorial explosion \and Markov bases \and origin-destination matrix \and population synthesis \and spatial interaction models}

\section{Introduction}\label{sec:intro}

Agent-based models (ABMs) are becoming increasingly popular policy-making tools in areas such as epidemic and transportation modelling \citep{bonabeau2002}. The emergent structure arising from ABM simulations relies on the quality of the underlying agent population's demographic and socioeconomic attributes. In transportation ABMs, such as MATSim \citep{axhausen2016}, simulated travel patterns are predominantly governed by the location where agent activities take place (e.g. working, shopping). The trips between activities are summarised in origin-destination matrices (ODMs), which are often either partially or not available a priori. Therefore, \textit{population synthesis} is performed to create artificial agents whose attributes (e.g. workplace location) have the same summary statistics as those described by population averages (e.g. regional job availability). Location choice synthesis translates to reconstructing integer-valued origin-destination matrices whose margins are summary statistics. To this end, coarse/aggregate agent activity surveys by geographical region and activity type are mainly leveraged \citep{fournier2021}. This is because fine-grained individual/disaggregate profiles are scarce and fragmented due to privacy and/or data acquisition cost reasons. Therefore, a discrepancy arises between the spatial resolutions of the data and latent states. Inferring individual agent trips subject to population summary statistics necessitates the exploration of a combinatorial choice space. The size of this space induces identifiability issues since a unique set of agent location choices consistent with the data cannot be recovered. 

A downsampling approach of sampling individual choices is computationally infeasible for any real-world application. Assuming that there are $M$ agents with $L$ available location choices, then computing the likelihood of the aggregate data given individual model parameters requires summing over $L^M$ possible location configurations, many of which are inconsistent with the data. Computational and identifiability issues can be alleviated by appropriately constraining the discrete latent space. The problem of exploring a constrained combinatorial agent state space is pertinent to any agent-based inference setting where the latent state is discrete. 

Although discrete choice models \citep{train2009} are popular candidates for disaggregating agent location choices, they cannot encode aggregate statistic constraints. Therefore, they either accrue errors when reconstructing ODMs or lead to factorially growing rejection rates \citep{desalvo2016} when forced to adhere to discrete constraints in a rejection-type scheme. A suite of greedy optimisation algorithms such as iterative proportional fitting \citep{bishop2007} and combinatorial optimisation \citep{voas2000a} were employed to assimilate summary statistic constraints in continuous and discrete spaces, respectively. These methods suffer from poor convergence to local optima, yielding solutions heavily dependent on good initialisations. Moreover, operating in a continuous probability/intensity space requires an additional sampling step to discretise the ODM, such as stochastic rounding \citep{croci2022}. This is an ad-hoc treatment of the problem and produces errors. The unidentifiable nature of disaggregating agent choices from aggregate data calls for uncertainty quantification in order to give practitioners the ability to interrogate and rank the sampled ODMs according to their probability.

Probabilistic methods have overcome some of the aforementioned limitations \citep{farooq2013,sun2016}, but remain approximate since they operate in the continuous intensity/probability space. In the case of location choice synthesis, ODMs are equivalent to two-way contingency tables of two categorical variables (e.g. origin residential population and destination workforce population) and the joint distribution of the two variables is explored using Gibbs sampling. Table marginal probabilities are elicited by normalising the discrete summary statistics. This approximation incurs information loss and may cause marginal class imbalances in high dimensional tables \citep{fournier2021}, meaning a growing divergence between ground truth and sampled marginal distributions. In addition, partially available data cannot be accommodated in a principled manner and unreasonable conditional independence assumptions are imposed.

\newlength{\oldintextsep}
\setlength{\oldintextsep}{\intextsep}
\setlength\intextsep{0pt}
\newlength{\oldcolumnsep}
\setlength{\oldcolumnsep}{\columnsep}
\setlength{\columnsep}{5pt}
\begin{wraptable}[20]{r}{0.65\textwidth}
    \centering
    \setlength{\tabcolsep}{4pt}
    \begin{tabular}{|
      >{\centering\arraybackslash}m{0.11\textwidth}
      >{\centering\arraybackslash}m{0.13\textwidth}
      >{\centering\arraybackslash}m{0.07\textwidth}
      >{\centering\arraybackslash}m{0.04\textwidth}
      >{\centering\arraybackslash}m{0.04\textwidth}
      >{\centering\arraybackslash}m{0.15\textwidth}
    |}
      \input{./tex_tables/summaries/contributions.tex}
    \end{tabular}
    \captionsetup{belowskip=0pt,aboveskip=0pt,width=.95\linewidth}
    \caption{Comparison of our method's capabilities against previous works. Agent choices are described by a discrete table ($\mathbf{\textcolor[HTML]{1E88E5}{T}}$) or a continuous intensity ($\boldsymbol{\textcolor[HTML]{FFC20A}{\Lambda}}$). Subscripts define summary statistics: the row and column sums/margins are indexed by ($\cdot,+), (+,\cdot)$, respectively. The cell universe $\mathcal{X}\supseteq \mathcal{X}'$ contains table/intensity indices of an $ I\times J$ matrix.}
    \label{tab:contributions}
\end{wraptable} 

The work of \citep{carvalho2014} endeavoured to address these two problems by adopting a Bayesian paradigm that operates directly on the discrete table space. However, neither the most efficient proposal mechanism nor the available intensity structure were exploited. Instead, a Metropolis-Hastings (MH) scheme for sampling contingency tables was employed that proposes jumps of size at most one in $\mathcal{O}($\# origins $\times$ \# destinations $)$, causing poor mixing in high-dimensional tables. Furthermore, the author argued for a hierarchical construction that jointly learns the constrained discrete ODM and the underlying intensity function. In doing so, they attempted to leverage a family of log-non-linear intensity models known as \textit{spatial interaction models} (SIMs) \citep{wilson1971}. SIMs incorporate summary statistic constraints directly in the continuous intensity space. Despite this effort, a log-linearity assumption was imposed on the SIM to simplify parameter inference. Also, the known dynamics of competition between destination locations \citep{dearden2015} were ignored, effectively stripping SIMs of all their embedded structure. Moreover, additional data were required to calibrate the intensity function, such as seed matrices, which are seldom available, as opposed to regularly observed data on the economic utility of travelling to destination locations. The works of \citep{ellam2018} and \citep{gaskin2023} alleviated this problem by constructing a physics-driven log-non-linear SIM intensity prior. However, both approaches operated strictly in the continuous intensity space and could not explore the discrete table space where population synthesis is performed.

\subsection{Contributions}

Our paper focuses on reconstructing origin-destination agent trip matrices summarising residence-to-workplace location choices. We offer an upsampling Bayesian approach that jointly samples from the discrete table ($\mathbf{\textcolor[HTML]{1E88E5}{T}}$) and continuous intensity ($\boldsymbol{\textcolor[HTML]{FFC20A}{\Lambda}}$) spaces for agent location choice synthesis. Our framework seamlessly assimilates any type of aggregate summary statistic as a constraint, which in turn regularises the space of admissible disaggregate/individual agent choices. We demonstrate an improved reconstruction and coverage of a partially observed origin-destination matrix summarising agent trips from residential to workplace locations in Cambridge, UK.

Contrary to the previous work, we perform a Gibbs step and sample tables in $\mathcal{O}($\# destinations$)$ leveraging the Markov basis machinery in \citep{diaconis1998} to design a Markov Chain Monte Carlo (MCMC) scheme with proposals that allow arbitrarily large jumps in table space without any accept/reject step. Hence, we bypass the problem of marginal distribution imbalances by respecting the exact margin frequencies rather than marginal distributions. We employ SIMs to understand the behavioural mechanism of aggregate location choice in continuous intensity space and relax previously adopted log-linearity assumptions on the intensity model. In the same fashion as \citep{ellam2018,gaskin2023}, we account for the stochastic dynamics of competition between destinations governing agent location choices and enforce an interpretable structure in the SIM intensity prior. A summary of our framework's capabilities relative to the previous works is depicted in Table \ref{tab:contributions}.

\section{Problem setup}\label{sec:problem_setup}

Consider $M$ agents that travel from $I$ origins to $J$ destinations to work. Let the expected number of trips (intensity) of agents between origin $i$ and destination $j$ be denoted by $\Lambda_{ij}$. The residential population in each origin (row sums) is equal to 
\begin{equation}\label{eq:orig_demand_constraint}
    \Lambda_{i+} = \sum_{j=1}^{J} \Lambda_{ij}, \;\;\;\; i=1,\dots,I,
\end{equation}
while the working population at each destination (column sums) is 
\begin{equation}\label{eq:dest_demand_constraint}
    \Lambda_{+j} = \sum_{i=1}^{I} \Lambda_{ij}, \;\;\;\; j=1,\dots,J.
\end{equation}
We assume that the total origin and destination demand are both conserved:
\begin{equation}\label{eq:total_constraint}
    M = \Lambda_{++} = \sum_{i=1}^{I} \Lambda_{i+} = \sum_{j=1}^{J} \Lambda_{+j}.
\end{equation}
This construction defines a totally constrained SIM. The demand for destination zones depends on the destination's attractiveness denoted by $\mathbf{w} \coloneqq (w_1,\dots, w_J) \in \mathbb{R}_{>0}^{J}$. Let the log-attraction be $\mathbf{x} \coloneqq \log(\mathbf{w})$. Between two destinations of similar attractiveness, agents are assumed to prefer nearby zones. Therefore, a cost matrix $\mathbf{C} = (c_{i,j})_{i,j=1}^{I,J}$ is introduced to reflect travel impedance. These two assumptions are justified by economic arguments \cite{pooler1994}. The maximum entropy distribution of agent trips subject to the total number of agents being conserved is derived by maximising
\begin{equation}\label{eq:objective_function}
    \mathcal{E}(\boldsymbol{\Lambda}) = \sum_{i=1}^{I}\sum_{j=1}^J -\Lambda_{ij}\log(\Lambda_{ij}) + \zeta  \left(\sum_{i,j}^{I,J} \Lambda_{ij} - M\right) + \alpha \sum_{j=1}^J x_j\Lambda_{ij} - \beta \sum_{i,j}^{I,J} c_{ij}\Lambda_{ij},
\end{equation}
which yields a closed-form expression for the trip intensity:
\begin{equation}\label{eq:max_entropy_trip}
    \Lambda_{ij} = \frac{\Lambda_{++}\exp(\alpha x_j -\beta c_{ij})}{\sum_{k,m}^{I,J} \exp(\alpha x_m-\beta c_{km})},
\end{equation}
where $\alpha,\beta$ control the two competing forces of attractiveness and deterrence. A higher $\alpha$ relative to $\beta$ characterises a preference over destinations with higher job availability, while the contrary indicates a predilection for closer workplaces. The destination attractiveness $\mathbf{w}$ is governed by the Harris-Wilson \citep{harris1978} system of $J$ coupled ordinary differential equations (ODEs):

\begin{equation}\label{eq:harris_wilson_ode}
    \frac{\dif w_j}{\dif t} = \epsilon w_j \left( \Lambda_{+j} - \kappa w_j + \delta  \right), \; \; \mathbf{w}(0) = \mathbf{w}',
\end{equation}
where $\kappa>0$ is the number of agents competing for one job, $\delta>0$ is the smallest number of jobs a destination can have and $\Lambda_{+j}(t) - \kappa w_j(t)$ is the net job capacity in destination $j$. A positive net job capacity translates to a higher economic activity (more travellers than jobs) and a boost in local employment, and vice versa. In equilibrium, the $J$ stationary points of the above ODE can be computed using

\begin{equation}\label{eq:harris_wilson_stationary_points}
    \kappa w_j - \delta = \frac{\Lambda_{++} w_j^{\alpha}}{\sum_{k,m}^{I,J} w_k^{\alpha}\exp(-\beta c_{km})} \sum_{i=1}^{I}\exp(-\beta c_{ij}).
\end{equation}
The value of $\kappa$ can be elicited by summing the above equation over destinations, which yields
\begin{equation}\label{eq:kappa}
    \kappa = \frac{\delta J + \Lambda_{++}}{\sum_{j=1}^J w_j},
\end{equation}
while $\delta$ corresponds to the case when no agent travels to destination $j'$ ($\Lambda_{+j'} = 0$), i.e.
\begin{equation}\label{eq:delta}
    \delta = \kappa \min_{j}\{w_j\}
\end{equation}

A stochastic perturbation of ~\ref{eq:harris_wilson_ode} incorporates uncertainty in the competition dynamics emerging from the randomness of agents' choice mechanisms. This gives rise to the Harris-Wilson stochastic differential equation (SDE) for the time evolution of log destination attraction $\mathbf{x}$
\begin{equation}\label{eq:harris_wilson_sde_log}
    \dif \mathbf{x} = - \epsilon^{-1} \nabla V(\mathbf{x})\dif t  + \sqrt{2\gamma^{-1}}\dif \mathbf{B}_t,  \;\;\;\; \mathbf{x}(0)=\mathbf{x_0},
\end{equation}
where the potential function $V(\mathbf{x})$ in the drift term is equal to 
\begin{equation}
 \epsilon^{-1}V(\mathbf{x}) = \underbrace{-\alpha^{-1} \sum_{i=1}^{I} O_i \log \left( \sum_{j=1}^{J} \exp (\alpha x_j - \beta c_{ij}) \right)}_{\text{utility potential}} + \underbrace{\kappa \sum_{j=1}^{J} \exp(x_{j})}_{\text{cost potential}} - \underbrace{\delta \sum_{j=1}^{J} x_j}_{\text{additional potential}}, 
 \label{eq:potential_function}
\end{equation}
and $\boldsymbol{\theta} = (\alpha,\beta)$ is the free parameter vector. 
The steady-state distribution of ~\ref{eq:harris_wilson_sde_log} is shown in \citep{ellam2018} to be the Boltzmann-Gibbs measure
\begin{align*}
    p(\mathbf{x}\vert\boldsymbol{\theta}) &= \frac{1}{Z(\boldsymbol{\theta})} \exp\left(-\gamma V_{\boldsymbol{\theta}}(\mathbf{x})\right) \numberthis \label{eq:boltzmann_gibbs_measure} \\
    Z(\boldsymbol{\theta}) &\coloneqq \int_{\mathbb{R}^J} \exp\left(-\gamma V_{\boldsymbol{\theta}}(\mathbf{x})\right) \dif\mathbf{x}. \numberthis \label{eq:boltzmann_gibbs_measure_normalising_constant}
\end{align*}
The observed data $\mathbf{y}$ are assumed to be noisy pertubartions of $\mathbf{x}$, where the error between the two satisfies $\log(\mathbf{e}) \sim \mathcal{N}(\mathbf{0},\sigma_d^2\mathbf{I})$, that is
\begin{equation}\label{eq:noise_model}
    \log(\mathbf{y}) = \mathbf{x} + \log(\mathbf{e}).
\end{equation}

\setlength{\oldintextsep}{\intextsep}
\setlength\intextsep{0pt}
\begin{wrapfigure}[14]{r}{0.3\textwidth}
    \centering
    \includegraphics[width=0.25\textwidth,keepaspectratio]{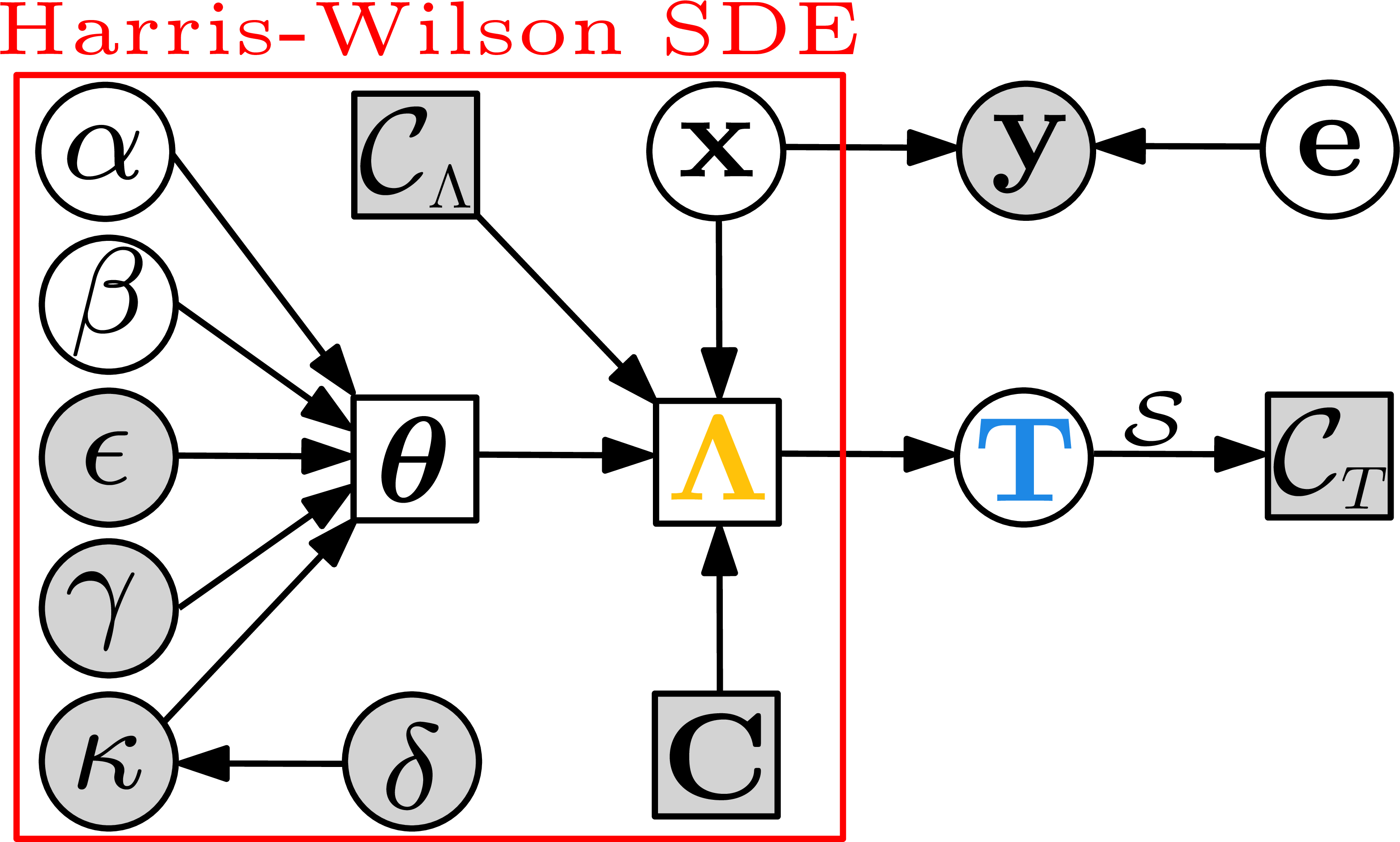}
    \captionsetup{width=.9\linewidth} 
    \caption{Plate diagram of our modelling framework. Rectangular and circular nodes are deterministic and random variables, respectively. Shaded nodes correspond to conditioned quantities.}
    \label{fig:plate_diagram}
\end{wrapfigure}

We introduce a data augmentation step to perform inference at the higher resolution origin-destination table space of agent trips as depicted in Figure \ref{fig:plate_diagram}. Assume that the $I\times J$ discrete contingency table $\mathbf{T}$ summarising the number of agents living in location $i$ and working in location $j$ is Poisson distributed:
\begin{equation}\label{eq:poisson}
    T_{ij}\sim \text{Poisson}\left(\Lambda_{ij}(\mathbf{x},\mathbf{\theta})\right),
\end{equation}
where the $T_{ij}$'s are conditionally independent given the $\Lambda_{ij}$'s. The contingency table inherits constraint ~\ref{eq:total_constraint}. These hard-coded constraints can be viewed as noise-free data on the discrete table space. We abbreviate the vector of row sums, column sums and the scalar total of $\mathbf{T}$ by $\mathbf{T_{\cdot +}}$, $\mathbf{T_{+ \cdot}}$ and $T_{++}$, respectively. Note that $\mathbf{T}$ uniquely determines the rest of the aforementioned random variables and $T_{++} = \Lambda_{++}$. Moreover, the distribution of $\mathbf{x}$ in ~\ref{eq:boltzmann_gibbs_measure} coupled with a prior on $\boldsymbol{\theta}$ jointly induces a prior over the intensity function $\boldsymbol{\Lambda}$.

Performing inference in a discrete higher-resolution table space circumvents challenges associated with enforcing summary statistic constraints in the continuous intensity space. First, the doubly constrained intensity (see Table \ref{tab:contributions}) admits solutions retrieved only through an iterative procedure that converges to poor local optima without any quantification of uncertainty, since the physical model in \eqref{eq:harris_wilson_sde_log} becomes redundant. Second, maximising \eqref{eq:objective_function} subject to individual cell constraints induces discontinuities in the $\Lambda$ space prohibiting SIM parameter calibration. To avoid dealing with discontinuities, a fully observable table is required, which is seldom available and defeats the purpose of ODM reconstruction. Alternatively, more parameters can be introduced, which entails identifiability problems as the number of free parameters becomes $\mathcal{O}(I+J)$ instead of $\mathcal{O}(J)$.  Moreover, augmenting $\mathcal{C}_{\Lambda}$ to match $\mathcal{C}_{T}$ strengthens the dependence between $\mathbf{T}\vert \boldsymbol{\Lambda},\mathcal{C}$ and $\mathbf{y}\vert \mathbf{x}, \mathcal{C}$. As a result, constraints are implicitly weighted (hard $\mathcal{C}_{\Lambda}$ and soft $\mathcal{C}_T$ constraints), which inflicts identifiability issues in $\boldsymbol{\Lambda}$.

\section{Discrete Table Inference}\label{sec:discrete_table_inference}

Let the set of table indices (cells) be $\mathcal{X} = \{(i,j) : 1\leq i\leq I, 1\leq j \leq J \}$ such that $T(x)=T_{ij}$ is the table value of cell $x=(i,j)\in \mathcal{X}$. For any subset $\mathcal{X}_k \subseteq \mathcal{X}$ let $S_k:\mathcal{X}_k \to \mathbb{N}^{I+J}$ be a bijective function that maps every cell $x \in \mathcal{X}_k$ to the $(I+J)$-dimensional binary vector with the $i$-th and $(I+j)$-th entries equal to one and the rest being zero. Define $\mathcal{S}_k: \mathcal{T} \to \mathbb{N}^{I+J}$ to be the summary statistic operator applying a uniquely defined $S_k(\cdot)$ to a table $\mathbf{T}\in \mathcal{T}$ over cells $\mathcal{X}_k'\subseteq \mathcal{X}$ such that $\mathcal{S}_k(\mathbf{T}') = \sum_{x\in\mathcal{X}} \mathbf{T}'(x) S_k(x)$. The ordered collection\footnote{Summary statistics are arranged in increasing order of cell set sizes $\vert\mathcal{X}_k\vert$. \label{foot:ordered_collection}} of summary statistic operators $\bigl\{\mathcal{S}_1(\mathbf{T}'),\dots, \mathcal{S}_K(\mathbf{T}')\bigl\}$ is abbreviated by $\mathbfcal{S}(\mathbf{T}')$. Define a collection of discrete summary statistics $\mathcal{C}_{T} = \bigl\{\mathbf{s}_1,\dots,\mathbf{s}_K\bigl\}$ expressed as constraints on table space, where each $\mathbf{s}_k$ is a realisation of $\mathcal{S}_k$. We leverage the same convention to define continuous constraints $\mathcal{C}_{\Lambda}$ in the intensity space. The union of table and intensity constraints is summarised by $\mathcal{C}$. We sometimes refer to $\mathcal{C}_T$ by $\mathcal{C}$ to avoid notation clutter. In Table \ref{tab:contributions} the singly constrained ODM model corresponds to a given $\mathcal{C}$ as opposed to singly constrained tables and intensities that map to $\mathcal{C}_T$ and $\mathcal{C}_{\Lambda}$, respectively. Equivalently, constrained models are defined by combinations of constrained tables and intensities.

\begin{definition}\label{def:admissibility}
Consider an ordered\footref{foot:ordered_collection} collection of constraints $\mathcal{C}_T$ and table summary statistics operators $\mathbfcal{S}$ with associated functions $\mathbf{S}$. A table $\mathbf{T}'$ is $\mathcal{C}_T$-\textit{admissible} if and only if its summary statistics satisfy all the constraints in $\mathcal{C}_T$, i.e. $\mathcal{S}_k(\mathbf{T}') = \mathbf{s}_k \in \mathcal{C}_T \; \forall \; k=1,\dots, K$.
\end{definition}

We denote the function space of all $\mathcal{C}_T$-admissible contingency tables (origin-destination matrices) of dimension $\dim(\mathbf{T}) = I\times J$ by $\mathcal{T}_{\mathcal{C}_T} = \{\mathbf{T} \in \mathcal{T}: \mathbfcal{S}(\mathbf{T}) = \mathcal{C}_T\}$ and drop the dependence on $T$ for notational convenience. This space contains all agent location choices consistent with the aggregate summary statistics $\mathcal{C}_T$. The set $\mathcal{T}_{\mathcal{C}_k}$ contains all tables that when applied $\mathcal{S}_k$ over cells $\mathcal{X}_k$ satisfy the $k$-th constraint of $\mathcal{C}_T$. In the rest of the paper we set $\mathcal{C}_{\Lambda}=\Bigl\{\Lambda_{++}\Bigl\}$ unless otherwrise stated. Our goal is to sample from $P(\mathbf{T},\mathbf{x},\boldsymbol{\theta}\vert\mathcal{C},\mathbf{y})$\footnote{We denote probability distributions over discrete, continuous and mixed discrete-continuous supports by $\mathbb{P}, p, P$, respectively.}, where
\begin{equation}\label{eq:full_posterior}
P(\mathbf{T},\underbrace{\mathbf{x},\boldsymbol{\theta}}_{\boldsymbol{\Lambda}}\vert\mathcal{C},\mathbf{y}) \propto \mathbb{P}(\mathbf{T}\vert \mathbf{x},\boldsymbol{\theta},\mathcal{C})p(\mathbf{y}\vert\mathbf{x})p(\mathbf{x}\vert\boldsymbol{\theta},\mathcal{C})p(\boldsymbol{\theta})
\end{equation}

We achieve this by devising a Metropolis-Hastings-within-Gibbs scheme to sample from $\mathbb{P}(\mathbf{T}\vert\mathbf{x},\boldsymbol{\theta},\mathcal{C})$, $p(\mathbf{x}\vert\boldsymbol{\theta},\mathbf{T},\mathbf{y})$ and $p(\boldsymbol{\theta}\vert\mathbf{x},\mathbf{T},\mathbf{y})$. The conditional samplers for $\mathbf{x}$ and $\boldsymbol{\theta}$ have acceptance ratios similar to those in \citep{ellam2018} and equal to
\begin{align*}
    p(\mathbf{x}',\mathbf{m}'\vert\mathbf{x},\mathbf{m},\boldsymbol{\theta},\mathbf{T},\mathcal{C},\mathbf{y}) &= 
    \min\left(1, \frac{\mathbb{P}(\mathbf{T}\vert \mathbf{x}',\boldsymbol{\theta},\mathcal{C})p(\mathbf{y}\vert\mathbf{x}')\exp\left(-H_{\theta}(\mathbf{x}')\right)}{\mathbb{P}(\mathbf{T}\vert \mathbf{x},\boldsymbol{\theta},\mathcal{C})p(\mathbf{y}\vert\mathbf{x})\exp\left(-H_{\theta}(\mathbf{x})\right)}\right),\numberthis \label{eq:x_acceptance_ratio} \\
     p(\boldsymbol{\theta}'\vert\boldsymbol{\theta},\mathbf{x},\mathbf{m},\mathbf{T},\mathcal{C},\mathbf{y}) &= \min\left(1,\frac{ \mathbb{P}(\mathbf{T}\vert \mathbf{x},\boldsymbol{\theta}',\mathcal{C})\exp\left(-\gamma V_{\boldsymbol{\theta}'}(\mathbf{x}) \right)Z(\boldsymbol{\theta})p(\boldsymbol{\theta}')}{ \mathbb{P}(\mathbf{T}\vert \mathbf{x},\boldsymbol{\theta},\mathcal{C})\exp\left(-\gamma V_{\boldsymbol{\theta}}(\mathbf{x}) \right)Z(\boldsymbol{\theta}')p(\boldsymbol{\theta})  } \right),\numberthis \label{eq:theta_acceptance_ratio}
\end{align*}
where $H_{\theta}(\mathbf{x}') = -\gamma V_{\boldsymbol{\theta}}(\mathbf{x}') - 1/2 \vert\mathbf{m}'\vert^2$ is the Hamiltonian of state $\mathbf{x}'$ with associated momentum $\mathbf{m}'$. Although a singly constrained intensity can be leveraged here, enforcing hard constraints through $\mathcal{C}_{\Lambda}$ and potentially different soft constraints through $\mathbb{P}(\mathbf{T}\vert\mathbf{x},\boldsymbol{\theta},\mathcal{C})$ would cause identifiability issues in $\mathbf{x}$. We aim to provide a general construction for joint table and intensity inference and employ singly constrained SIMs only when $\mathcal{C}_T=\{\mathbf{T}_{\cdot +}\}$. In the following exposition, we show that the type of summary statistic data available determines whether the constrained table distribution $\mathbb{P}(\mathbf{T}\vert\mathbf{x},\boldsymbol{\theta},\mathcal{C})$ can be sampled directly or indirectly through Markov Chain Monte Carlo (MCMC). 

\subsection{Tractable constrained table sampling}
In this section, we offer closed-form contingency table sampling. Without loss of generality, assume that only one of the two table margins is known; namely $\mathbf{T}_{+\cdot}$ (singly constrained table). Then, any subset $\mathcal{C}_{T}$ of the universe of summary statistic constraints $ \Bigl\{ T_{++}, \mathbf{T_{+\cdot}},\bigl\{ \mathbf{T}_{\mathcal{X}_l} \vert \mathcal{X}_l\subseteq \mathcal{X}, l \in \mathbb{N} \bigl\} \Bigl\}$ yields a closed-form posterior table marginal as shown in Table \ref{tab:contributions}. By the construction in \eqref{eq:poisson}, the case for $\mathcal{C}_T=\emptyset$ is equivalent to unconstrained table sampling conditioned on an intensity model, which in our case is a SIM. Cell constraints $\mathcal{C}_T=\bigl\{ \mathbf{T}_{\mathcal{X}_l} \vert \mathcal{X}_l\subseteq \mathcal{X}, l \in \mathbb{N} \bigl\}$ can be seamlessly incorporated in an unconstrained table without violating the posterior's tractability. Furthermore, leveraging that $\mathbf{T}$ uniquely determines both $T_{++}$ and $\mathbf{T}_{+\cdot}$ and applying Bayes' rule it follows that the models with $\mathcal{C}_T=\{T_{++}\}$ and $\mathcal{C}_T=\{\mathbf{T}_{+\cdot}\}$ yield Multinomial and product Multinomial distributions, respectively (See full derivations in Appendix \ref{app:table_posterior_marginals}). Equivalently,

\begin{equation}\label{eq:Multinomial}
    \mathbb{P}(\mathbf{T}\vert\boldsymbol{\Lambda},T_{++}) = \prod_{i,j}^{I,J} \left( \frac{T_{++}!}{T_{ij}!} \left(\frac{\Lambda_{ij}}{\Lambda_{++}} \right)^{T_{ij}} \right),
\end{equation}
and
\begin{equation}\label{eq:product_Multinomial}
    \mathbb{P}(\mathbf{T}\vert\boldsymbol{\Lambda},\mathbf{T_{\cdot +}}) = \prod_{i,j}^{I,J} \left( \frac{T_{i+}!}{T_{ij}!} \left(\frac{\Lambda_{ij}}{\Lambda_{i+}} \right)^{T_{ij}} \right).
\end{equation}
We obtain $N$ independent samples $\mathbf{T}^{(1:N)}$ from \eqref{eq:poisson}, \eqref{eq:Multinomial} \eqref{eq:product_Multinomial} in closed-form. Samples from the Poisson and product Multinomial distributions can be drawn in parallel. We note that the space complexity of table sampling is $\mathcal{O}(IJ)$ while the time complexity for \eqref{eq:poisson},  \eqref{eq:Multinomial} \eqref{eq:product_Multinomial} is $\mathcal{O}(IJ)$, $\mathcal{O}(1)$ and $\mathcal{O}(I)$, respectively. Moreover, coupling either constraint $T_{++}$ or $\mathbf{T}_{\cdot+}$ with cell constraints leaves the target distribution unchanged but shrinks its support. Hence, the available table margin is updated by subtracting the value of fixed cell constraints from margin statistics and performing inference on the free cells. We present the joint intensity and table sampling algorithm for tractably constrained tables in Algorithm ~\ref{alg:tractable_table_sampling_algorithm}.

\begin{algorithm}
	\caption{Metropolis-within-Gibbs MCMC sampling algorithm for tractably constrained tables.}%
    \label{alg:tractable_table_sampling_algorithm}
	\begin{algorithmic}[1]
        \State \textbf{Inputs}: $\mathbf{C}$, $\mathcal{C}$, $\mathbf{y}$, $N$.
        \State \textbf{Outputs}: $\mathbf{x}^{(1:N)},\boldsymbol{\theta}^{(1:N)},\text{sign}\left(\boldsymbol{\theta}^{(1:N)}\right),\mathbf{T}^{(1:N)}$.
        \State Initialise $\mathbf{x}^{(0)}$, $\boldsymbol{\theta}^{(0)}$, $\mathbf{T}^{(0)}$.
	    \For{$n \in \{1,\dots,N\}$}
	    \State Sample $\mathbf{x}^{(n)}\vert\boldsymbol{\theta}^{(n-1)},\mathbf{T}^{(n-1)}$ using Hamiltonian Monte Carlo \citep{neal2011} with acceptance ~\ref{eq:x_acceptance_ratio}.
	    \State Sample $\boldsymbol{\theta}^{(n)}\vert\mathbf{x}^{(n)},\mathbf{T}^{(n-1)}$ using Random Walk Metropolis-Hastings with acceptance ~\ref{eq:theta_acceptance_ratio}.
	    \State Construct intensity $\boldsymbol{\Lambda}^{(n)}$ using $\mathbf{x}^{(n)}, \boldsymbol{\theta}^{(n)}$ using ~\ref{eq:max_entropy_trip}.
	    \State Sample tables (in parallel) using the relevant closed-form distribution (~\ref{eq:poisson},~\ref{eq:Multinomial},~\ref{eq:product_Multinomial}) $\mathbf{T}^{(n)}\sim \mathbb{P}(\mathbf{T}\vert\boldsymbol{\Lambda}^{(n)},\mathcal{C})$. 
	    \EndFor
	\end{algorithmic}
\end{algorithm}

\subsection{Intractable constrained table sampling}
In this section, we introduce an MCMC scheme for sampling tables subject to any subset of the power set $\mathcal{P}\biggl( \Bigl\{ \mathbf{T}_{\cdot+},\mathbf{T}_{+\cdot},\bigl\{ \mathbf{T}_{\mathcal{X}_l} \vert \mathcal{X}_l\subseteq \mathcal{X}, l \in \mathbb{N} \bigl\} \Bigl\} \biggl)$ excluding those subsets contained in the constraint universe of the previous section. By conditioning on both table margins and leveraging the conditional distributions of $\mathbf{T}_{\cdot+} \vert T_{++}, \boldsymbol{\Lambda}$ and $T_{++} \vert \boldsymbol{\Lambda}$, the induced conditional distribution becomes Fisher's non-central multivariate hypergeometric \citep{agresti2002}:
\begin{equation}\label{eq:fishers_hypergeometric}
    \mathbb{P}(\mathbf{T}\vert\boldsymbol{\Lambda},\mathbf{T_{. +}},\mathbf{T_{\cdot +}}) \propto \frac{\prod_{i=1}^I T_{i+}! \prod_{j=1}^J T_{+j}!}{T_{++}!\prod_{i,j=1}^{I,J} T_{ij}!} \prod_{i,j=1}^{I,J} \left(\frac{\Lambda_{ij} \Lambda_{++}}{\Lambda_{i+} \Lambda_{+j}}\right)^{T_{ij}},
\end{equation}
where $\omega_{ij} = \frac{\Lambda_{ij}\Lambda_{++}}{\Lambda_{i+}\Lambda_{+j}}$ is called the odds ratio and encodes the strength of dependence between row $i$ and column $j$. Complete independence is achieved if and only if $\omega_{ij}=1$. Our choice of intensity model encodes this dependence in the travel cost matrix $\mathbf{C}$. Origin-destination independence is achieved if and only if the travel cost's effect on destination choice is irrelevant ($\beta=0$). Moreover, the normalising constant of \eqref{eq:fishers_hypergeometric} is a partition function defined over the support of all tables satisfying the conditioned margins and can't be efficiently computed by direct enumeration. In Appendix \ref{app:chu_vandermonde_theorem} we prove an extension of Chu-Vandermonde's convolution theorem for Multinomial coefficients \citep{belbachir2014} that facilitates computation of the normalising constant in $\mathcal{O}(1)$. In particular, we show that the following identity holds:
\begin{equation}\label{eq:chu_vandermonde_identity}
    \binom{T_{++}}{T_{+1}\dots T_{+J}} \prod_{i,j}^{J} \omega_{ij}^{T_{+j}} = \sum_{\mathcal{S}\left(\mathbf{T}\right)=\mathcal{C}_T} \prod_{i,j}^{I,J} \binom{T_{+j}}{T_{1j}\dots T_{Ij}} \omega_{ij}^{T_{ij}},
\end{equation}
where $\binom{T_{+j}}{T_{1j}\dots T_{Ij}} = \frac{T_{+j}!}{T_{1j}!\dots T_{Ij}!}$ is the Multinomial coefficient. Shrinking the $\mathcal{T}_{\mathcal{C}}$ space using elements of the constraint universe above requires a Markov Basis (MB) MCMC sampling scheme \citep{diaconis1998} due to the intractability of the induced table posterior. 

\subsubsection{Markov Basis MCMC}
We construct a $\mathcal{C}_T$-admissible table for initialising Markov Basis MCMC using a suite of greedy deterministic algorithms, such as iterative proportional fitting \citep{bishop2007}. We concoct a proposal mechanism on $\mathcal{T}_{\mathcal{C}}$ as follows.

\begin{definition}\label{def:null_admissibility}
    A \textit{null-admissible} table $\mathbf{T}$ is a $\mathcal{C}_T$-admissible table with $\mathcal{C}_T \subseteq \{\mathbf{T}_{+\cdot},\mathbf{T}_{\cdot+}\}$ and $\forall \; \mathbf{s} \in \mathcal{C}_T$ it follows that $\mathbf{s} = \mathbf{0}$.
\end{definition}

\begin{definition}\label{def:markov_basis}: A \textit{Markov basis} is a set of table moves $\mathbf{f}_1,\dots,\mathbf{f}_L: \mathcal{X} \to \mathbb{Z}$ that satisfy the following conditions: 
    \begin{enumerate}
        \item $\mathbf{f}_l$ is a null-admissible table for $1\leq l\leq L$ and 
        \item for any two $\mathcal{C}_T$-admissible $\mathbf{T},\mathbf{T}'$ there are $\mathbf{f}_{l_1},\dots,\mathbf{f}_{l_A}$ with $\eta_l \in \mathbb{N}$ such that $\mathbf{T}'=\mathbf{T}+\sum_{m=1}^A \eta_l \mathbf{f}_{l_m}$ and $\mathbf{T}+\sum_{m=1}^a \eta_l \mathbf{f}_{l_m} \geq 0\;$ for $\; 1\leq a \leq A$.
    \end{enumerate}
\end{definition}

Condition (i) guarantees that all proposed moves do not modify the summary statistics in $\mathcal{C}_T$, while condition (ii) ensures that there exists a path between any two tables such that any table member of the path is $\mathcal{C}_T$-admissible. The collection of constraints $\mathcal{C}_T$ generates a Markov basis $\mathcal{M}$. When $I\times J$ tables satisfy both row and column margins, $\mathcal{M}$ consists of functions $\mathbf{f}_1,\dots, \mathbf{f}_L$ such that $\forall \; x = (i_1,j_1), x' = (i_2,j_2) \in \mathcal{X}$ with $i_{1}\neq i_{2}, j_{1}\neq j_{2},$
\begin{equation}\label{eq:markov_basis_two_margins}
    \mathbf{f}_l(x) = \begin{cases} 
      \eta & \text{if} \; x = (i_1,j_1) \; \text{or} \; x = (i_1,j_2) \\
      -\eta & \text{if} \; x = (i_2,j_2)\; \text{or} \; x = (i_2,j_1) \\
      0 & \text{otherwise}
   \end{cases}
\end{equation}
The case for coupling individual cell constraints with table margins requires a minor modification. Let $\mathcal{X}' \subseteq \mathcal{X}$ and $\mathcal{C}'_T$ be the individual cell admissibility criteria. Then, $\mathcal{M}$ is updated to exclude all basis functions $\mathbf{f}_l$ with $\mathbf{f}_l(x)\neq 0 \;\; \forall \; x \in \mathcal{X}'$. Moreover, $\mathcal{C}_T$ is revised so that $\forall \; \mathbf{s} \in \mathcal{C}_T$, $\mathbf{s}' \in \mathcal{C}'_T$, $\mathbf{s}$ is updated to $\mathbf{s}-\mathbf{s}'$ at every $x \in \mathcal{X}'$. In other words, the constrained cell values are deducted from the rest of the summary statistic constraints in $\mathcal{C}_T$. A Markov Basis Markov chain (MBMC) can now be constructed.

\begin{proposition}\label{prop:markov_basis_mcmc} 
    (Adapted from \citep{diaconis1998}): Let $\mathcal{\mu}$ be a probability measure on $\mathcal{T}_{\mathcal{C}}$. Given a Markov basis $\mathcal{M}$ that satisfies ~\ref{def:markov_basis}, generate a Markov chain in $\mathcal{T}_{\mathcal{C}}$ by sampling $l$ uniformly at random from $\{1,\dots,L\}$. Consider the Markov Basis Metropolis-Hastings (MB-MH) and Gibbs (MB-Gibbs) proposals:
    \begin{enumerate}
        \item MB-MH: Let $\eta \in \{-1,1\}$ and choose $\eta$ from this set with probability $\frac{1}{2}$ independent of $l$. If the chain is at $\mathbf{T} \in \mathcal{T}_{\mathcal{C}}$ it will move to $\mathbf{T}' = \mathbf{T} + \eta \mathbf{f}_{l}$ with probability $$ \min \biggl\{ \frac{\mu(\mathbf{T} + \eta \mathbf{f}_{l})}{\mu(\mathbf{T})},1  \biggl\}$$ provided $\mathbf{T}' \geq 0$. In all other cases, the chain stays at $\mathbf{T}$. 
        \item MB-Gibbs: Let $\eta \in \mathbb{Z}$. If the chain is at $\mathbf{T} \in \mathcal{T}_{\mathcal{C}}$, determine the set of $\eta $ such that $\mathbf{T} + \eta \mathbf{f}_l \geq 0$. Choose $$ \mathbb{P}(\eta) \propto \prod_{x\in \{x\in \mathcal{X}: \mathbf{f}_l(x) \neq 0 \}} \frac{1}{\mu\biggl(n(x) + \eta \mathbf{f}_l(x) \biggl)} $$ and move to $\mathbf{T}' = \mathbf{T} + \eta \mathbf{f}_l \geq 0$.
    \end{enumerate}
    In both cases an aperiodic, reversible, connected Markov chain in $\mathcal{T}_{\mathcal{C}}$ is constructed with stationary distribution proportional to $\mu(\mathbf{T})$.
\end{proposition}

The proof of Proposition ~\ref{prop:markov_basis_mcmc} is provided in \citep{diaconis1998}. Theoretical guarantees of Markov Basis MCMC convergence on $\mathcal{T}_{\mathcal{C}}$ show that the MB-MH scheme in ~\ref{prop:markov_basis_mcmc} mixes slowly and is not scalable to high-dimensional $I\times J$ tables for large $T_{++}$. Instead, a Gibbs sampler can be constructed as detailed in the same proposition (MB-Gibbs). 

In doubly constrained tables, $\eta$ is distributed according to Fisher's non-central hypergeometric distribution for $2\times 2$ tables. The derivation of this result is provided in in Corollary \ref{corollary:markov_basis_proposal_distribution} of Appendix \ref{app:table_posterior_marginals}. The overhead of generating $\mathcal{M}$ for any constrained table is at most $\mathcal{O}(I^2J^2)$ in both time and space. This overhead can be easily overcome by amortising the construction of $\mathcal{M}$ prior to sampling. The sampling procedure for a constrained model with an intractable table marginal distribution and underlying SIM intensity model is summarised in Algorithm ~\ref{alg:intractable_sampling_algorithm}. The time complexity of proposing a move in $\mathcal{T}_{\mathcal{C}}$ is $\mathcal{O}(1)$ and $\mathcal{O}(\max \bigl\{\max(\mathbf{s}) \; \big\vert \; \mathbf{s} \in \mathcal{C} \bigl\})$ for MB-MH and MB-Gibbs, respectively. The corresponding space complexities are both $\mathcal{O}(IJ)$.

\begin{algorithm}
	\caption{Metropolis-within-Gibbs Markov Basis MCMC sampling algorithm for intractably constrained tables.}
	\label{alg:intractable_sampling_algorithm}
	\begin{algorithmic}[1]
        \State \textbf{Inputs}: $\mathbf{C}$, $\mathcal{C}$, $\mathbf{y}$,  $\mathcal{M}$, $\mu$, $N$.
        \State \textbf{Outputs}: $\mathbf{x}^{(1:N)},\boldsymbol{\theta}^{(1:N)},\text{sign}\left(\boldsymbol{\theta}^{(1:N)}\right),\mathbf{T}^{(1:N)}$.
        \State Initialise $\mathbf{x}^{(0)}$, $\boldsymbol{\theta}^{(0)}$, $\mathbf{T}^{(0)}$.
	    \For{$m \in \{1,\dots,M\}$}
        \State Sample $\mathbf{x}^{(n)}\vert\boldsymbol{\theta}^{(n-1)},\mathbf{T}^{(n-1)}$ using Hamiltonian Monte Carlo \citep{neal2011} with acceptance ~\ref{eq:x_acceptance_ratio}.
	    \State Sample $\boldsymbol{\theta}^{(n)}\vert\mathbf{x}^{(n)},\mathbf{T}^{(n-1)}$ using Random Walk Metropolis-Hastings with acceptance ~\ref{eq:theta_acceptance_ratio}.
	    \State Construct intensity $\boldsymbol{\Lambda}^{(n)}$ using $\mathbf{x}^{(n)}, \boldsymbol{\theta}^{(n)}$ using ~\ref{eq:max_entropy_trip}.
	    \State Sample $l$ uniformly at random from $\{1,\dots,L\}$.
        \State Find the valid $\eta$ support yielding $\mathcal{C}_T$-admissible tables.
        \State Use MB-Gibbs in case 2 of ~\ref{prop:markov_basis_mcmc} to sample valid $\eta$ with specified $\mu$.
        \State Obtain $\mathbf{T}^{(n)} = \mathbf{T}^{(n-1)} + \eta \mathbf{f}_{l}$.
	    \EndFor
	\end{algorithmic}
\end{algorithm}

The curse of dimensionality prohibits the use of any standard convergence diagnosis techniques, such as the Gelman and Rubin criterion \citep{gelman1992}. Therefore, we employ the $l_1$ norm to empirically assess the convergence of sample summary statistics and establish convergence in probability. Furthermore, we assume the underlying intensity function is known a priori, which acts as a ground truth. In the case of Fisher's non-central hypergeometric distribution, exact moments are not available \citep{mccullagh2019}. These are approximated by the moments of a product Multinomial kernel derived in Appendix \ref{app:table_posterior_marginals}.

\section{Experimental Results and Discussion}\label{sec:experimental_results_and_discussion}

We showcase table sampling convergence results based on a fixed synthetic intensity across different numbers of origins $I$, destinations $J$ and agents $M=T_{++}$. Figure \ref{fig:table_size_agent_number_curse_of_dimensionality} depicts empirical convergence rates based on a total of $10^3$ chains each run for $10^3$ steps. Sparse tables (\ref*{subplot:dim_33x33_total_100_experiment_title_row_margin}) induce multimodal distributions in $\mathcal{T}_{\mathcal{C}}$ and mix slowly compared to their dense counterparts (\ref*{subplot:dim_33x33_total_5000_experiment_title_row_margin}). Convergence is decelerated more by a larger number of agents rather than higher table dimensionality. The number of agents grows as fast as the diameter of the chain's state space and bounds the number of MCMC steps required to reach the stationary distribution. This observation agrees with the theoretical bounds obtained in \citep{diaconis1998}, although the latter bounds are derived based on a uniform measure over $\mathcal{T}_{\mathcal{C}}$ explored using MB-MH. Despite this discrepancy, theoretical results provide an upper bound for our case of direct sampling, as evidenced by Figure \ref{fig:mixing_times_mcmc_proposals}. Direct sampling from the closed-form table posterior achieves the fastest convergence, and we use it to benchmark against Markov basis MCMC. Any doubly constrained table can be explored using either MB-MH (\ref*{subplot:proposal_degree_one_experiment_title_both_margins}) or MB-Gibbs (\ref*{subplot:proposal_degree_higher_experiment_title_both_margins}). Encoding additional constraints in $\mathcal{T}$ to contract the posterior entails the overhead of using MCMC, introducing a tradeoff between convergence rate and distribution contraction in the presence of more summary statistic constraints $\mathcal{C}_{T}$.

\begin{figure}[!ht]
  \centering
  \begin{minipage}[t]{0.5\textwidth}
    \centering
      \captionsetup{width=.9\linewidth}
      \input{./tex_figures/intensity_and_table_convergence/mean_table_relative_l_1_norm_vs_iteration_comparison_by_table_dim__table_total_K_1000.tex}
      \caption{$l_1$ error norm of $\mathbb{E}[\mathbf{T}\vert\mathbf{y},\mathbf{T}_{\cdot +}]$ across table sizes $\dim(\mathbf{T})$ and number of agents $T_{++}$ using Algorithm \ref{alg:tractable_table_sampling_algorithm}. Convergence is slower for sparse tables (\ref*{subplot:dim_33x33_total_100_experiment_title_row_margin}) that induce multimodal distributions. As $T_{++}$ grows (\ref*{subplot:dim_2x3_total_5000_experiment_title_row_margin},\ref*{subplot:dim_33x33_total_5000_experiment_title_row_margin}) convergence is decelerated by a factor inversely proportional to the table size, which agrees with the theoretical bounds established in \citep{diaconis1998}.}%
      \label{fig:table_size_agent_number_curse_of_dimensionality}
  \end{minipage}%
  \begin{minipage}[t]{0.5\textwidth}
    \centering
    \captionsetup{width=.9\linewidth}
      \input{./tex_figures/intensity_and_table_convergence/mean_table_relative_l_1_norm_vs_iteration_comparison_by_proposal_K_1000.tex}
      \caption{$l_1$ error norm of $\mathbb{E}[\mathbf{T}\vert\mathbf{y},\mathcal{C}_T]$ for a $33\times 33$ table with $5000$ agents in the singly (\ref*{subplot:proposal_direct_sampling_experiment_title_row_margin},\ref*{subplot:proposal_degree_one_experiment_title_row_margin},\ref*{subplot:proposal_degree_higher_experiment_title_row_margin}) and doubly (\ref*{subplot:proposal_degree_one_experiment_title_both_margins}, \ref*{subplot:proposal_degree_higher_experiment_title_both_margins}) constrained tables. MB-Gibbs has a substantially faster convergence rate than MB-MH and mixes reasonably slower compared to direct sampling. Ground truth averages $g(\boldsymbol{\Lambda})$ are approximate for doubly constrained tables (\ref*{subplot:proposal_degree_one_experiment_title_both_margins}, \ref*{subplot:proposal_degree_higher_experiment_title_both_margins}).}%
      \label{fig:mixing_times_mcmc_proposals}
  \end{minipage}
\end{figure}
\begin{figure*}[!ht]
  \begin{subfigure}[b]{0.475\textwidth}
    \centering
  \input{./tex_figures/2d_sample_projections/table_2d_isomap_label_by_experiment_title_gamma_burnin10000_thinning50_euclidean_distance_nearest_neighbours100.tikz}
    \label{fig:lower_dim_embedding_table}
  \end{subfigure}
  \begin{subfigure}[b]{0.475\textwidth}
    \centering    \input{./tex_figures/2d_sample_projections/intensity_2d_isomap_label_by_experiment_title_gamma_burnin10000_thinning50_euclidean_distance_nearest_neighbours100.tikz}
    \label{fig:lower_dim_embedding_intensity}
  \end{subfigure}
  \caption{Visualisation of the table (left) and intensity (right) samples projected in 2D using T-distributed stochastic neighbour embedding \citep{hinton2002}. Samples are coloured by the constraint sets in Table \ref{tab:metrics} for low (\ref*{subplot:intensity_low_noise}), high (\ref*{subplot:intensity_high_noise}) and variable (\ref*{subplot:variable_noise}) noise regimes. The ground truth table (\ref*{subplot:ground_truth_table}) is better covered by the discrete table posterior regardless of $\mathcal{C}$, and the table distribution becomes increasingly concentrated around the ground truth table in light of more data $\mathcal{C}_T$. Intensity samples are weakly informed through $\mathcal{C}_{\Lambda}$ and $\mathbb{P}(\mathbf{T}\vert \mathbf{x},\boldsymbol{\theta},\mathcal{C})$ and more distant from the ground truth.}
    \label{fig:lower_dim_embedding_samples}
\end{figure*}

Furthermore, we present a large-scale application of discrete ODM reconstruction to Cambridge commuting patterns from residence to workplace locations, using the ODM models in Table \ref{tab:contributions}. The precise experimental setup mimics that of \citep{ellam2018} and is provided in the Supporting information. In light of new summary statistics $\mathcal{C}_T$ (e.g. \ref*{subplot:table_both_margins_high_noise}, \ref*{subplot:table_both_margins_permuted_cells_10percent_high_noise},\ref*{subplot:table_both_margins_permuted_cells_20percent_high_noise}), the table posterior contracts and its high mass region concentrates around the ground truth table (\ref*{subplot:ground_truth_table}), as shown in Figure \ref{fig:lower_dim_embedding_samples}. The fact that the low-noise table samples (\ref*{subplot:table_row_margin_low_noise},\ref*{subplot:table_both_margins_low_noise},\ref*{subplot:table_both_margins_permuted_cells_10percent_low_noise},\ref*{subplot:table_both_margins_permuted_cells_20percent_low_noise}) are nearby their high-noise counterparts (\ref*{subplot:table_row_margin_high_noise},\ref*{subplot:table_both_margins_high_noise},\ref*{subplot:table_both_margins_permuted_cells_10percent_high_noise},\ref*{subplot:table_both_margins_permuted_cells_20percent_high_noise}) indicates a more dominant effect of the table likelihood on the posterior relative to that of the intensity SDE prior, which enforces the confidence in our reconstructed ODM. The intensity samples of \citep{gaskin2023} (\ref*{subplot:intensity_NeuralABM_row_margin_low_noise},\ref*{subplot:intensity_NeuralABM_row_margin_high_noise},\ref*{subplot:intensity_NeuralABM_row_margin_learned_noise}) have the highest variance amongst the sampled intensities due to the random initialisations of the Neural ODE solver in \citep{gaskin2023}. Despite this, the intensity distributions in \citep{ellam2018} and \citep{gaskin2023} have insufficient $\mathcal{C}_{\Lambda}$ constraints and a higher divergence from the ground truth table region than table samples. Our intensity samples are also distant from the ground truth table (\ref*{subplot:intensity_JointTableSIMLatentMCMC_row_margin_low_noise},\ref*{subplot:intensity_JointTableSIMLatentMCMC_both_margins_low_noise},\ref*{subplot:intensity_JointTableSIMLatentMCMC_both_margins_permuted_cells_10percent_low_noise},\ref*{subplot:intensity_JointTableSIMLatentMCMC_both_margins_permuted_cells_20percent_low_noise},\ref*{subplot:intensity_JointTableSIMLatentMCMC_row_margin_high_noise},\ref*{subplot:intensity_JointTableSIMLatentMCMC_both_margins_high_noise},\ref*{subplot:intensity_JointTableSIMLatentMCMC_both_margins_permuted_cells_10percent_high_noise},\ref*{subplot:intensity_JointTableSIMLatentMCMC_both_margins_permuted_cells_20percent_high_noise}) because they are informed strongly by $\mathcal{C}_{\Lambda}$ and weakly by $\mathcal{C}_{T}$ (See Figure \ref{fig:plate_diagram}), where the former set is smaller than the latter. 

\begin{table}[!ht]
  \centering
  \begin{adjustbox}{max width=\textwidth,center}
    \begin{tabular}{
      V{3.0}
      >{\centering\arraybackslash}m{0.23\textwidth}V{3.0}
      >{\centering\arraybackslash}m{0.18\textwidth}
      >{\centering\arraybackslash}m{0.05\textwidth}
      >{\centering\arraybackslash}m{0.03\textwidth}
      >{\centering\arraybackslash}m{0.05\textwidth}
      >{\centering\arraybackslash}m{0.05\textwidth}
      >{\centering\arraybackslash}m{0.16\textwidth}
      >{\centering\arraybackslash}m{0.14\textwidth}
      >{\centering\arraybackslash}m{0.18\textwidth}
      >{\centering\arraybackslash}m{0.16\textwidth}
      V{3.0}
      }
      \input{./tex_tables/summaries/table_metrics.tex}
    \end{tabular}
  \end{adjustbox}
  \caption{Ground truth table validation metrics comparing our method against \citep{ellam2018} and \citep{gaskin2023} in the continuous intensity and discrete table levels across noise regimes $\gamma$ and constraint sets $\mathcal{C}$. Our method achieves the best error-coverage tradeoff in the doubly and 20\% cell constrained ODM as well as the best reconstruction errors (MDB, SRMSE), data fit (SSI) and ground truth coverage (CP) across all ODM models, as indicated by \mbox{\large $*$}. As the size of $\mathcal{C}$ increases, the ground truth table is better reconstructed (reduced SRMSE, increased SSI) and covered by the table posterior. In the cases of total and singly constrained ODMs, we attain a similar SSI and coverage to \citep{gaskin2023}. We note that the last three ODM models cannot be handled by \citep{ellam2018} or \citep{gaskin2023}.}
  \label{tab:metrics}
  \end{table}

The ODM validation results summarised in Table \ref{tab:metrics} affirm that reasoning at the discrete table level accomplishes greater error reductions and enhanced ground truth coverage. Data fitness and posterior prediction errors are computed using the Sorensen similarity index (SSI), standardised root mean square error (SRMSE) and Markov Basis distance (MBD). Uncertainty quantification is evaluated based on the coverage probability (CP) of ground truth table cells contained in the $99\%$ highest posterior mass (HPM) region. We elucidate each of these metrics in the Supporting information. The best error-coverage tradeoff, lowest SRMSE, MBD, and highest SSI are attained in the doubly and $20\%$ cell constrained model due to it having the richest constraint set $\mathcal{C}$. Our doubly constrained models account for an SRMSE reduction of $16\%$ relative to the singly constrained model while sustaining an acceptable ground truth coverage equal to approximately $80\%$. The apparent increase in the mean intensity SRMSE across all doubly constrained models potentially alludes to the SIM's lack of expressivity. This may be because $\mathcal{C}_T$ and $\mathbf{y}$ give rise to conflicting SIM parameter configurations in the limit of large $\mathcal{C}_T$. The MBD decrease in the growth of $\mathcal{C}$ indicates that the expected upper bound on the number of Markov Basis moves required to exactly match $\mathbf{T}^{\mathcal{D}}$ is reduced. In the totally and singly constraint models, our table posterior mean matches or outperforms the intensity mean of \citep{ellam2018} and \citep{gaskin2023} in terms of data fit (SSI) and SRMSE. The highest ground truth cell coverage probability ($94\%$) is achieved by the most relaxed table, namely the unconstrained table, but entails a high bias. A lower SRMSE ($0.67$ instead of $0.85$) is attained by the intensity field of the totally constrained model in \citep{gaskin2023}, at the expense of a coverage probability drop from $94\%$ to $85\%$ and a discretisation error accrued for population synthesis.

Our framework's benefits also extend to SIM parameter estimation. In Figure \ref{fig:posterior_mean_x_predictions} we show that the log destination attraction prediction $R^2$ increases for larger constraint sets $\mathcal{C}_T$ from $0.77$ to $0.84$. This allows us to explain the evolved destination employment by informing the data-generating process through $\mathcal{C}$ instead of increasing the diffusivity of the SDE prior in \eqref{eq:boltzmann_gibbs_measure}. Therefore, we mitigate the identifiability issues of the multimodal $\boldsymbol{\theta}$ posterior emerging in the high noise regime. The $\mathbf{x}$ predictions are further improved in the high noise regime ($R^2=0.99$) compared to the low noise counterpart ($R^2=0.84$), which favours the hypothesis of a stochastic growth in destination employment. In the high noise regime, unbiased estimators of $\boldsymbol{\theta}$ are devised based on a more disperse SDE prior on $\boldsymbol{\Lambda}$ \ref{eq:boltzmann_gibbs_measure}. Increased prior diffusivity steers the $\mathbf{x}$ posterior marginal towards a larger region of plausible SDE solutions in the vicinity of $\mathbf{y}$, which improves the quality of $\mathbf{x}$ predictions. Additionally, we recover the $\mathbf{x}$ and $\boldsymbol{\theta}$ posterior marginals obtained in \citep{ellam2018} at a fraction of additional computational cost.

\begin{figure}[!ht]
  \centering
  \begin{subfigure}{.5\textwidth}
    \centering
\input{./tex_figures/log_destination_attraction_predictions/table_69x13_JointTableSIMLatentMCMC_log_destination_attraction_predictions_burnin_10000_thinning_1_N_100000_low_noise.tex}
      \label{fig:posterior_mean_x_predictions_low_noise}
  \end{subfigure}%
  \begin{subfigure}{.5\textwidth}
    \centering
\input{./tex_figures/log_destination_attraction_predictions/table_69x13_JointTableSIMLatentMCMC_log_destination_attraction_predictions_burnin_10000_thinning_1_N_100000_high_noise.tex}
    \label{fig:posterior_mean_x_predictions_high_noise}
  \end{subfigure}
  \caption{Posterior predictions of $\mathbb{E}[\mathbf{x}\vert\mathbf{y},\mathcal{C}]$ against observed log employment data $\mathbf{y}$ using Algorithms \ref{alg:tractable_table_sampling_algorithm} and \ref{alg:intractable_sampling_algorithm}. The high noise regime (right) achieves a more plausible data fit than its low noise counterpart (left) due to improved uncertainty quantification. This is attributed to the unbiased estimation of the normalising constant \eqref{eq:boltzmann_gibbs_measure_normalising_constant} obtained in the former case compared to biased estimation and subsequent collapse of the $\boldsymbol{\theta}$ posterior to a Dirac mass in the latter case. Enriching $\mathcal{C}_T$ from \ref*{subplot:type_JointTableSIMLatentMCMC_noise_regime_low_experiment_title_grand_total},\ref*{subplot:type_JointTableSIMLatentMCMC_noise_regime_low_experiment_title_row_margin} to \ref*{subplot:type_JointTableSIMLatentMCMC_noise_regime_low_experiment_title_both_margins},\ref*{subplot:type_JointTableSIMLatentMCMC_noise_regime_low_experiment_title_both_margins_permuted_cells_10percent},\ref*{subplot:type_JointTableSIMLatentMCMC_noise_regime_low_experiment_title_both_margins_permuted_cells_20percent} increases $R^2$, illustrating the added benefits of inference on a higher resolution table level.}%
  \label{fig:posterior_mean_x_predictions}
\end{figure}

In conclusion, performing population synthesis directly on the discrete high-resolution space of agent attributes bears tangible empirical benefits. These include improved reconstruction and coverage of the ground truth ODM, as well as table posterior contraction in the limit of constraint data $\mathcal{C}_T$. If population synthesis is not of interest, SIM parameters can be adequately estimated using competitive approaches such as \citep{gaskin2023}. Combining such optimisation methods with Markov Basis MCMC in a naive Bayes scheme can be promising, as it exploits the advantages of both optimisation and MCMC techniques. Regardless, the apparent shortcomings of SIMs call for a comparative study of various intensity model classes, such as discrete choice models \citep{train2009}. Finally, the multi-faceted nature of population synthesis opens up future avenues of research beyond ODM reconstruction, where more convoluted dependency structures can be exploited.

\section*{Software and Data}
Trip and employment data for Cambridge, UK are obtained from  \citep{ons_trip_data,ons_employment_data}. Individual home and work facility locations are extracted from\citep{osm_data}. Our codebase has been released on \href{https://github.com/YannisZa/ticodm}{https://github.com/YannisZa/ticodm}.


\printbibliography

\section*{Appendix}\label{sec:appendices}

\appendix

\section{Table posterior marginals}\label{app:table_posterior_marginals}
In this section, we derive all conditional table posteriors outlined in Table \ref{tab:contributions}. Assume that the trip intensity $\boldsymbol{\Lambda}$ is known. 

\begin{remark}\label{remark:table_statistic_distributions}
    Assuming conditional independence of $T_{ij}$ and $T_{km}$ given $\boldsymbol{\Lambda} \;\; \forall \; i\neq k,\;j\neq m$ standard results \citep{bishop2007} dictate that the following hold:
    \begin{align*}
        \mathbf{T}_{\cdot +}\vert \boldsymbol{\Lambda} &\sim \text{Poisson}\left(\boldsymbol{\Lambda}_{\cdot +}\right) \\
        T_{++}\vert \boldsymbol{\Lambda} &\sim \text{Poisson}\left(\Lambda_{++}\right)
    \end{align*}
\end{remark}

\begin{corollary}\label{corollary:table_margins}
The random variable $\mathbf{T}_{\cdot +} \vert T_{++}, \boldsymbol{\Lambda}$ follows a Multinomial distribution with $T_{++}$ number of trials and $\frac{\boldsymbol{\Lambda}_{\cdot +}}{\Lambda_{++}}$ event probabilities.
\end{corollary}
\begin{proof}
    Apply Baye's rule and notice that $\mathbf{T}$ fully determines any member of $\mathcal{C}_T$. It follows that 
    \begin{align*}
        \mathbb{P}\left(\mathbf{T}_{\cdot +} \vert T_{++}, \boldsymbol{\Lambda}\right) &=  \prod_{i}^{I} \frac{\Lambda_{ij}^{T_{i+}}\exp(-\Lambda_{i+})}{T_{i+}!} \left( \frac{\Lambda_{ij}^{T_{++}}\exp(-\Lambda_{++})}{T_{++}!} \right)^{-1} \\
        &= \prod_{i}^{I} \frac{T_{++}!}{T_{i+}!} \left(\frac{\Lambda_{i+}}{\Lambda_{++}}\right)^{T_{i+}}.
    \end{align*}
    This is the kernel of a Multinomial distribution on a one-way contingency table $\mathbf{T}_{\cdot +}$ with the desired parameters.
\end{proof}

\begin{lemma}
Let $\mathcal{C}_T=\left\{T_{++}\right\}$. Then, $\mathbf{T} \vert \mathcal{C},\boldsymbol{\Lambda}$ is a Multinomially distributed random variable with $T_{++}$ number of trials and $\frac{\boldsymbol{\Lambda}}{\Lambda_{++}}$ event probabilities.
\end{lemma}
\begin{proof}
    Apply Baye's rule, notice that $\mathbf{T}$ fully determines any member of $\mathcal{C}_T$ and leverage Remark \ref{remark:table_statistic_distributions}. It follows that 
    \begin{align*}
        \mathbb{P}\left(\mathbf{T} \vert \mathcal{C},\boldsymbol{\Lambda}\right) &= \frac{\mathbb{P}\left(\mathbf{T}\vert \mathcal{C}_{\Lambda}, \boldsymbol{\Lambda}\right)}{\mathbb{P}\left(T_{++} \vert \mathcal{C}_{\Lambda}, \boldsymbol{\Lambda}\right)}\\
        &= \prod_{i,j}^{I,J} \frac{\Lambda_{ij}^{T_{ij}} \exp(-\Lambda_{ij})}{T_{ij}!} \left( \frac{\Lambda_{ij}^{T_{++}} \exp(-\Lambda_{++})}{T_{++}!}  \right)^{-1} \\
        &= \prod_{i,j}^{I,J} \frac{T_{++}!}{T_{ij}!} \left(\frac{\Lambda_{ij}}{\Lambda_{++}}\right)^{T_{ij}},
    \end{align*}
    which resembles the kernel of a Multinomial distribution on a two-way contingency table $\mathbf{T}$ with the desired parameters.
\end{proof}

\begin{lemma}
    Let $\mathcal{C}_T=\left\{\mathbf{T}_{\cdot+}\right\}$. Then, $\mathbf{T} \vert \mathcal{C},\boldsymbol{\Lambda}$ is a random variable distributed according to a product Multinomial law with $\mathbf{T}_{\cdot+}$ number of trials and $\frac{\boldsymbol{\Lambda}}{\boldsymbol{\Lambda}_{\cdot +}}$ event probabilities for each table row.
\end{lemma}
\begin{proof}
    Apply Baye's rule and Corollary \ref{corollary:table_margins}. It follows that 
    \begin{align*}
        \mathbb{P}\left(\mathbf{T} \vert \mathcal{C},\boldsymbol{\Lambda}\right) &= \frac{\mathbb{P}\left(\mathbf{T}\vert \mathcal{C}_{\Lambda}, \boldsymbol{\Lambda}\right)}{\mathbb{P}\left(\mathbf{T}_{\cdot+} \vert T_{++},\mathcal{C}_{\Lambda}, \boldsymbol{\Lambda}\right)\mathbb{P}\left(T_{++} \vert \mathcal{C}_{\Lambda}, \boldsymbol{\Lambda}\right)}\\
        &= \prod_{i,j}^{I,J} \frac{\Lambda_{ij}^{T_{ij}} \exp(-\Lambda_{ij})}{T_{ij}!} \left( \prod_{k}^{I} \frac{T_{++}!}{T_{k+}!} \left(\frac{\Lambda_{k+}}{\Lambda_{++}}\right)^{T_{k+}} \frac{\Lambda_{ij}^{T_{++}} \exp(-\Lambda_{++})}{T_{++}!} \right)^{-1} \\
        &= \prod_{i,j}^{I,J} \frac{T_{i+}!}{T_{ij}!} \left(\frac{\Lambda_{ij}}{\Lambda_{i+}}\right)^{T_{ij}},
    \end{align*}
    which is proportional to the desired distribution.
\end{proof}

\begin{lemma}\label{lemma:fisher_hypergeometric}
    Let $\mathcal{C}_T=\left\{\mathbf{T}_{\cdot+},\mathbf{T}_{+\cdot}\right\}$. Then, $\mathbf{T} \vert \mathcal{C},\boldsymbol{\Lambda}$ is a random variable distributed according to Fisher's non-central hypergeometric law with $\mathbf{T}_{\cdot+}, \mathbf{T}_{\cdot+}, T_{++}$ table margins and odds ratios equal to $\boldsymbol{\omega} = \frac{\Lambda_{++} \boldsymbol{\Lambda}}{\boldsymbol{\Lambda}_{\cdot +}\boldsymbol{\Lambda}_{+\cdot}}$.
\end{lemma}
\begin{proof}
    Apply Baye's rule, Remark \ref{remark:table_statistic_distributions}, Corollary \ref{corollary:table_margins} and exploit the conditional independence between $\mathbf{T}_{\cdot +}$ and $\mathbf{T}_{+\cdot}$ given $\boldsymbol{\Lambda}$. It follows that 
    \begin{align*}
        \mathbb{P}\left(\mathbf{T} \vert \mathcal{C},\boldsymbol{\Lambda}\right) &\propto \frac{\mathbb{P}\left(\mathbf{T}\vert \mathcal{C}_{\Lambda}, \boldsymbol{\Lambda}\right)}{\mathbb{P}\left(\mathbf{T}_{\cdot+} \vert T_{++},\mathcal{C}_{\Lambda}, \boldsymbol{\Lambda}\right)\mathbb{P}\left(\mathbf{T}_{+\cdot} \vert T_{++},\mathcal{C}_{\Lambda}, \boldsymbol{\Lambda}\right)\mathbb{P}\left(T_{++} \vert \mathcal{C}_{\Lambda}, \boldsymbol{\Lambda}\right)}\\
        &= \prod_{i,j}^{I,J} \frac{\Lambda_{ij}^{T_{ij}} \exp(-\Lambda_{ij})}{T_{ij}!} \left( \prod_{k}^{I} \frac{T_{++}!}{T_{k+}!} \left(\frac{\Lambda_{k+}}{\Lambda_{++}}\right)^{T_{k+}} \prod_{m}^{J} \frac{T_{++}!}{T_{+m}!} \left(\frac{\Lambda_{+m}}{\Lambda_{++}}\right)^{T_{+m}} \frac{\Lambda_{ij}^{T_{++}} \exp(-\Lambda_{++})}{T_{++}!} \right)^{-1} \\
        &= \prod_{i,j}^{I,J} \frac{T_{i+}!T_{+j}!}{T_{ij}!T_{++}!} \left(\frac{\Lambda_{ij}\Lambda_{++}}{\Lambda_{i+}\Lambda_{+j}}\right)^{T_{ij}},
    \end{align*}
    by virtue of Remark \ref{remark:table_statistic_distributions}. This resembles the kernel of Fisher's non-central hypergeometric distribution on a two-way contingency table $\mathbf{T}$ with the desired parameters.
\end{proof}

\begin{corollary}\label{corollary:markov_basis_proposal_distribution}
In doubly constrained tables, the Markov Basis step size proposal distribution $\mathbb{P}(\eta)$ is proportional to Fisher's non-central hypergeometric distribution for $2\times 2$ contingency tables with table margins $T_{i_1+}, T_{i_2+}, T_{+j_1}, T_{+j_2}$ and odds ratio $\omega_{ij} = \frac{\Lambda_{i_1j_1}\Lambda_{i_2j_2}}{\Lambda_{i_1j_2}\Lambda_{i_2j_1}}$. 
\end{corollary}
\begin{proof}
    Without loss of generality consider an arbitrary Markov Basis move illustrated in Figure \ref{fig:markov_basis}. All entries in the table remain constant except for entries at cells $(i_1,j_1),(i_1,j_2),(i_2,j_1),(i_2,j_2)$. Let the updated table be $\mathbf{T}'$. Expand the distribution in Lemma \ref{lemma:fisher_hypergeometric} and simplify terms depending only on $\mathbf{T},\boldsymbol{\Lambda}$:
    \InsertBoxR{0}{\begin{minipage}{0.3\linewidth}
        \centering
        \includegraphics[width=0.6\linewidth]{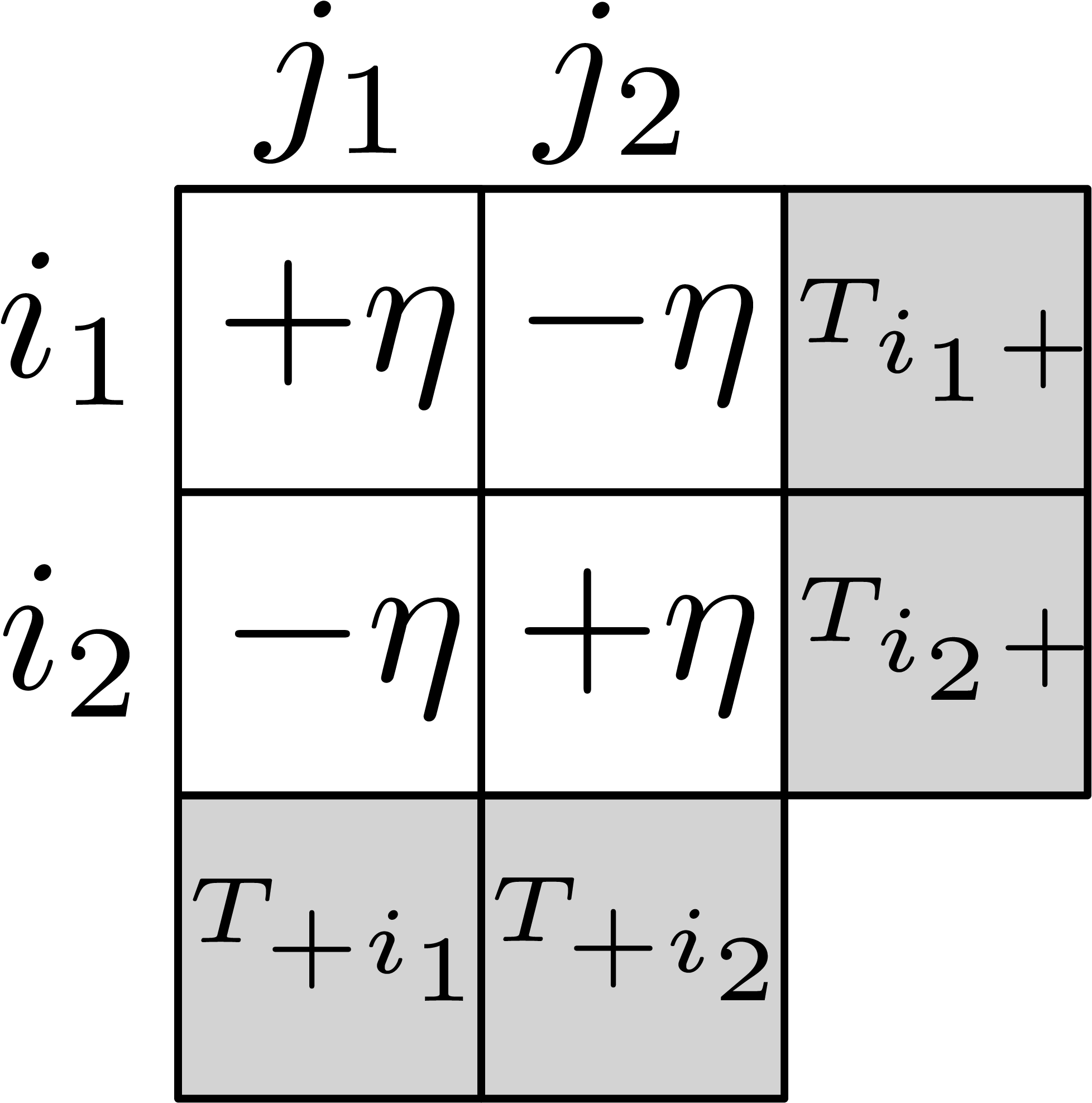}
        \captionsetup{width=.9\linewidth}
        \captionof{figure}{Markov basis move on the space of $I\times J$ contingency tables with $\mathcal{C}_T=\{\mathbf{T}_{\cdot+},\mathbf{T}_{+\cdot}\}$.}
        \label{fig:markov_basis}
    \end{minipage}%
    }[10]
    \begin{align*}
        \mathbb{P}\left(\eta\right) &= \frac{T_{i_1+}!T_{i_2+}!T_{+j_1}!T_{+j_2}!}{T_{i_1j_1}!T_{i_1j_2}!T_{i_2j_1}!T_{i_2j_2}!} 
        \left(\frac{\Lambda_{++}\Lambda_{i_1j_1}}{\Lambda_{i_1+}\Lambda_{+j_1}}\right)^{(T_{i_1j_1}+\eta)} 
        \left(\frac{\Lambda_{++}\Lambda_{i_2j_2}}{\Lambda_{i_2+}\Lambda_{+j_2}}\right)^{(T_{i_2j_2}+\eta)} \\
        & \qquad
        \left(\frac{\Lambda_{++}\Lambda_{i_2j_1}}{\Lambda_{i_2+}\Lambda_{+j_1}}\right)^{(T_{i_2j_1}-\eta)}
        \left(\frac{\Lambda_{++}\Lambda_{i_1j_2}}{\Lambda_{i_1+}\Lambda_{+j_2}}\right)^{(T_{i_1j_2}-\eta)} \\
        &\propto \frac{1}{T_{i_1j_1}!T_{i_1j_2}!T_{i_2j_1}!T_{i_2j_2}!} \left(\frac{\Lambda_{++}\Lambda_{i_1j_1}}{\Lambda_{i_1+}\Lambda_{+j_1}}\right)^{\eta} 
        \left(\frac{\Lambda_{++}\Lambda_{i_2j_2}}{\Lambda_{i_2+}\Lambda_{+j_2}}\right)^{\eta} 
        \left(\frac{\Lambda_{++}\Lambda_{i_2j_1}}{\Lambda_{i_2+}\Lambda_{+j_1}}\right)^{-\eta}
        \left(\frac{\Lambda_{++}\Lambda_{i_1j_2}}{\Lambda_{i_1+}\Lambda_{+j_2}}\right)^{-\eta} \\
        &= \frac{1}{T_{i_1j_1}!T_{i_1j_2}!T_{i_2j_1}!T_{i_2j_2}!} \left(\frac{\Lambda_{i_1j_1}\Lambda_{i_2j_2}}{\Lambda_{i_2j_1}\Lambda_{i_1j_2}}\right)^{\eta},
    \end{align*}
    which is the kernel of Fisher's non-central hypergeometric on a $2\times 2$ table with the desired parameters.
\end{proof}

\section{An extension to Chu Vandermonde's theorem for Multinomial coefficients}\label{app:chu_vandermonde_theorem}

    \begin{theorem}\label{thm:chu_vandermonde_extension}
    Let $\mathcal{T}_{\mathcal{C}}$ be the space of admissible tables satisfying $\mathcal{C}_T=\{\mathbf{T}_{+\cdot},\mathbf{T}_{\cdot +}\}$ with fixed odds ratios $\boldsymbol{\omega} \in \mathbb{R}_{\geq 0}$. Denote the subsets of $\mathcal{C}_T$-admissible tables with $\mathcal{C}_T=\{T_{++}\}$, $\mathcal{C}_T=\{\mathbf{T}_{+\cdot}\}$ and $\mathcal{C}_T=\{\mathbf{T}_{\cdot+}\}$ by $\mathcal{T}_{++}$, $\mathcal{T}_{+\cdot}$ and $\mathcal{T}_{\cdot+}$, respectively. Then, for any $\mathcal{C}_T$-admissible $I\times J$ table $\mathbf{T}$ the following statement holds:
    $$\binom{T_{++}}{T_{+1}\dots T_{+J}} \prod_{j}^{J} \omega_{+j}^{T_{+j}} = \sum_{\mathbf{T}\in \mathcal{T}_{+\cdot}} \prod_{j}^{J} \binom{T_{+j}}{T_{1j}\dots T_{Ij}} \prod_{i,j}^{I,J} \omega_{ij}^{T_{ij}}.$$
    This is an extension of the Chu-Vandermonde theorem for Multinomial coefficients \citep{belbachir2014} to polynomials with Multinomial coefficients.
    \end{theorem}
    \begin{proof}
        \allowdisplaybreaks
        We proceed with an algebraic proof. Let $[x]_{j=1}^{J}=1+x+\cdots+x^J$ be a polynomial of order $J$. By writing the left-hand side in polynomial form and expanding it we get
        \begin{align*}
            \sum_{\mathbf{T}_{+\cdot}\in \mathcal{T}_{++}} \underbrace{\left( \binom{T_{++}}{T_{+1}\dots T_{+J}} \prod_{j}^{J} \omega_{+j}^{T_{+j}} \right)}_{\text{LHS}} \prod_{m}^{J} x_{m}^{T_{+m}} 
            &= \sum_{\mathbf{T}_{+\cdot}\in \mathcal{T}_{++}} \binom{T_{++}}{T_{+1}\dots T_{+J}} \prod_{j}^{J} \left(\omega_{+j}x_j\right)^{T_{+j}} \\
            &= \left(\sum_{j}^J \omega_{+j}x_j\right)^{T_{++}} \\
            &= \prod_{j}^J \left(\sum_{j}^J \omega_{+j}x_j\right)^{T_{+j}} \\
            &= \prod_{j}^J \sum_{\mathbf{T}\in \mathcal{T}_{+j}} \binom{T_{+j}}{T_{1j}\dots T_{Ij}} \prod_{i}^{I} (\omega_{ij}x_j)^{T_{ij}} \\
            &= \prod_{j}^J \left(\sum_{\mathbf{T}\in \mathcal{T}_{+j}} \binom{T_{+j}}{T_{1j}\dots T_{Ij}} \prod_{i}^{I} \omega_{ij}^{T_{ij}}\right) x_j^{T_{+j}}\\
            &= \sum_{\mathbf{T}_{+\cdot}\in \mathcal{T}_{++}} \underbrace{\left(\sum_{\mathbf{T}\in \mathcal{T}_{+\cdot}} \prod_{j}^J \binom{T_{+j}}{T_{1j}\dots T_{Ij}} \prod_{i}^{I} \omega_{ij}^{T_{ij}}\right)}_{\text{RHS}} \prod_{j}^J x_j^{T_{+j}},
        \end{align*}
        where the exchange of product and sum in the last line is permitted due to the grouping of terms in the sum by column $\mathbf{T}_{\cdot j} \; \forall \; j=1, \dots, J$. The second and fourth equalities follow by direct application of the Multinomial theorem \citep{berge1971principles}. This completes the proof.
    \end{proof}

The Theorem above allows us to compute the normalising constant of Fisher's non-central hypergeometric distribution in Lemma \ref{lemma:fisher_hypergeometric}. We note that summing over the support $\mathcal{T}_{+\cdot}$ yields
\begin{align*}
    \frac{\prod_{i}^{I} T_{i+}!}{T_{++}!} \sum_{\mathbf{T} \in \mathcal{T}_{+\cdot}} \prod_{i,j}^{I,J} \frac{T_{+j}!}{T_{ij}!} \omega_{ij}^{T_{ij}} &= \frac{\prod_{i}^{I} T_{i+}!}{T_{++}!} \prod_{j}^{J} 
 \frac{T_{++}!}{T_{+j}!} \omega_{+j}^{T_{+j}} \\
 &= \prod_{i,j}^{I,J} \frac{T_{i+}!}{T_{+j}!} \omega_{+j}^{T_{+j}}.
\end{align*}
Subsequently, the normalised Fisher's non-central hypergeometric distribution is equal to 
\begin{equation}\label{eq:fisher_hypergeometric_normalised}
    \prod_{j}^{J} \frac{T_{+j}!T_{+j}!}{T_{++}!T_{ij}!} \prod_{i}^{I} \left(\frac{\omega_{ij}}{\omega_{+j}}\right)^{T_{ij}} =     \prod_{j}^{J} \frac{\binom{T_{+j}}{T_{1j}, \dots, T_{Ij}}}{\binom{T_{++}}{T_{+1}, \dots, T_{+J}}} \prod_{i}^{I} \left(\frac{\omega_{ij}}{\omega_{+j}}\right)^{T_{ij}},
\end{equation} 
which is similar to the kernel of a product Multinomial with $\mathbf{T}_{+\cdot}$ number of trials, $\frac{\boldsymbol{\omega}}{\boldsymbol{\omega}_{+\cdot}}$ event probabilities. We adopt this approximation for computing ground truth moments in our synthetic doubly constrained model experiments of Figure \ref{fig:mixing_times_mcmc_proposals}.

\section{Table posterior gradients}\label{app:table_posterior_gradients}

The Hamiltonian Monte Carlo update of $\mathbf{x}$ necessitates the use of the table log-likelihood gradients inside the leapfrog integrator. We list these gradients for the totally and singly constrained intensities with $\mathcal{C}_{\Lambda}=\{\Lambda_{++}\}$ and $\mathcal{C}_{\Lambda}=\{\boldsymbol{\Lambda}_{+\cdot}\}$, respectively. The total derivative of the table log-likelihood is equal to 
\begin{equation*}
    \frac{\partial \log\left(\mathbb{P}\left(\mathbf{T}\vert \mathcal{C},\boldsymbol{\Lambda}\right)\right)}{\partial x_m} = \sum_{i,j} \frac{\partial \log\left(\mathbb{P}\left(\mathbf{T}\vert \mathcal{C},\boldsymbol{\Lambda}\right)\right)}{\partial \Lambda_{ij}} \frac{\partial\Lambda_{ij}}{\partial x_m}.
\end{equation*}

\begin{remark}\label{remark:total_constrained_gradient}
    The totally constrained intensity $\Lambda_{ij} = \frac{\Lambda_{++} \exp\left(\alpha x_j - \beta c_{ij}\right)}{\sum_{k,m}^{I,J} \exp\left(\alpha x_m - \beta c_{km}\right)}$ has 
    \begin{align*}
        \frac{\partial \Lambda_{ij}}{\partial x_m} = \begin{cases}
            \alpha\Lambda_{ij}\left(1-\frac{\Lambda_{+j}}{\Lambda_{++}}\right) & \text{if} \; m=j, \\
            -\alpha \frac{\Lambda_{+m}}{\Lambda_{++}} & \text{if} \; \text{otherwise}.
        \end{cases}
    \end{align*}
\end{remark}

\begin{remark}\label{remark:singly_constrained_gradient}
    The singly constrained intensity $\Lambda_{ij} = \frac{\Lambda_{i+} \exp\left(\alpha x_j - \beta c_{ij}\right)}{\sum_{m}^{J} \exp\left(\alpha x_m - \beta c_{im}\right)}$ has 
    \begin{align*}
        \frac{\partial \Lambda_{ij}}{\partial x_m} = \begin{cases}
            \alpha\Lambda_{ij}\left(1-\frac{\Lambda_{ij}}{\Lambda_{i+}}\right) & \text{if} \; m=j, \\
            -\alpha \frac{\Lambda_{im}\Lambda_{ij}}{\Lambda_{i+}} & \text{if} \; \text{otherwise}.
        \end{cases}
    \end{align*}
\end{remark}

\begin{remark}\label{remark:odds_ratio_gradient}
    The odds ratio $\omega_{ij} = \frac{\Lambda_{ij}\Lambda_{++}}{\Lambda_{i+}\Lambda_{+j}}$ has
    \begin{equation*}
        \frac{\partial \omega_{km}}{\partial \Lambda_{ij}} = \begin{cases}
            \frac{\omega_{ij}}{\Lambda_{ij}} + \frac{\omega_{ij}}{\Lambda_{ij}} - \frac{\omega_{ij}}{\Lambda_{+j}} - \frac{\omega_{ij}}{\Lambda_{i+}} & \text{if} \; k=i,m=j, \\
            \frac{\omega_{kj}}{\Lambda_{++}} - \frac{\omega_{kj}}{\Lambda_{++}\Lambda_{+j}} & \text{if} \; k\neq i,m=j, \\
            \frac{\omega_{im}}{\Lambda_{++}} - \frac{\omega_{im}}{\Lambda_{++}\Lambda_{i+}} & \text{if} \; k=i,m\neq j, \\
            \frac{\omega_{km}}{\Lambda_{++}} & \text{if} \; \text{otherwise}.
        \end{cases}
    \end{equation*}
\end{remark}

\begin{remark}
    Let $\mathcal{C}_T = \emptyset$. Then, 
    $$\frac{\partial \log\left(\mathbb{P}\left(\mathbf{T}\vert \mathcal{C},\boldsymbol{\Lambda}\right)\right)}{\partial \Lambda_{ij}} = \frac{T_{ij}}{\Lambda_{ij}} - \Lambda_{i+}.$$
\end{remark}

\begin{remark}
    Let $\mathcal{C}_T=\left\{T_{++}\right\}$. Then, 
    $$\frac{\partial \log\left(\mathbb{P}\left(\mathbf{T}\vert \mathcal{C},\boldsymbol{\Lambda}\right)\right)}{\partial \Lambda_{ij}} = \frac{T_{ij}}{\Lambda_{ij}} - \frac{T_{++}}{\Lambda_{++}}.$$
\end{remark}

\begin{remark}
    Let $\mathcal{C}_T=\left\{\mathbf{T}_{\cdot +}\right\}$. Then, 
    $$\frac{\partial \log\left(\mathbb{P}\left(\mathbf{T}\vert \mathcal{C},\boldsymbol{\Lambda}\right)\right)}{\partial \Lambda_{ij}} = \frac{T_{ij}}{\Lambda_{ij}} - \frac{T_{+j}}{\Lambda_{+j}}.$$
\end{remark}

\begin{remark}
    Let $\mathcal{C}_T=\left\{\mathbf{T}_{\cdot +},\mathbf{T}_{+\cdot}\right\}$. Then, 
    \begin{equation*}
        \frac{\partial \log\left(\mathbb{P}\left(\mathbf{T}\vert \mathcal{C},\boldsymbol{\Lambda}\right)\right)}{\partial \Lambda_{ij}} = \sum_{k,m}^{I,J} \frac{\partial \log\left(\mathbb{P}\left(\mathbf{T}\vert \mathcal{C},\boldsymbol{\Lambda}\right)\right)}{\partial \omega_{km}} \frac{\omega_{km}}{\Lambda_{ij}},
    \end{equation*}
    where 
    \begin{equation*}
    \frac{\partial \log\left(\mathbb{P}\left(\mathbf{T}\vert \mathcal{C},\boldsymbol{\Lambda}\right)\right)}{\partial \omega_{km}} =  \frac{T_{ij}}{\omega_{ij}} - \frac{T_{+j}}{\omega_{+j}},
    \end{equation*}
    and the second derivative term is calculated using Remark ~\ref{remark:odds_ratio_gradient}.
\end{remark}

\section{Experimental setup}\label{app:experimental_setup}

In this section, we elucidate the experimental setup used for synthetic experiments and the large-scale application to commuting patterns in Cambridge, UK. 

\subsection{Synthetic experiments}

Regarding synthetic experiments with fixed $\boldsymbol{\Lambda}$, Figure \ref{fig:table_size_agent_number_curse_of_dimensionality} is produced by running $10^3$ instances of line 7 from Algorithm \ref{alg:tractable_table_sampling_algorithm} for $10^3$ steps. Synthetic tables with dimensions $2\times 3$, $33\times 33$ and a total number of agents $10^2$ and $5\times 10^3$ are created. The constraint set is $\mathcal{C}_T=\{\mathbf{T}_{\cdot+}\}$. Figure \ref{fig:mixing_times_mcmc_proposals} is run on the same synthetic datasets with $\mathcal{C}_T=\{\mathbf{T}_{\cdot+}\}$ and $\mathcal{C}_T=\{\mathbf{T}_{\cdot+},\mathbf{T}_{+\cdot}\}$ using lines 7 and 7-10 of Algorithms \ref{alg:tractable_table_sampling_algorithm} and \ref{alg:intractable_sampling_algorithm}, respectively. Each one of the $10^3$ MCMC instances is run for $10^4$ steps using both MB-MH and MB-Gibbs to compare convergence rates. The ground truth table average $g(\boldsymbol{\Lambda})$ is computed using the moments of the product Multinomial distribution with known probabilities. In doubly constrained models, this ground truth constitutes an approximation introduced in Appendix \ref{app:chu_vandermonde_theorem}, since the moments of Fisher's non-central hypergeometric cannot be computed in closed form.

\subsection{Real-world application to Cambridge, UK}

The Cambridge home-to-work trip dataset contains $69$ origins and $13$ destinations at lower super output area (LSOA) and middle super output area (MSOA) spatial resolutions, respectively. The total number of agents whose work trips are recorded is $33,704$. Job availability data from \citep{ons_trip_data,ons_employment_data} is leveraged as observed log destination attraction $\mathbf{y}$. A cost matrix of travel deterrence is computed as follows. Residence and workplace facilities and Cambridge's entire transportation network are extracted from \citep{osm_data}. Separate clusters of origin and destination facilities are created for each origin LSOA and destination MSOA with each cluster size being equal to $20$ facilities to limit computational cost. Dijkstra's shortest path algorithm \citep{dijkstra2022} is used to compute the distance between all pairwise origin-destination facility pairs. Spatial effects arising from unobserved facilities near the spatial boundary cause overestimation of travel costs from and to boundary LSOAs and MSOAs. To mitigate this effect, Ripley's k function \citep{dixon2001} is computed to estimate the spatial density of home and work facilities in a radius of $500$ meters around the boundary. Then, the cost matrix is multiplied by the percentage of facilities within both the $500$ meter radius and the modelled LSOAs and MSOAs of interest relative to the total number of facilities within the radius. The rationale behind this edge correction is that a larger proportion of unaccounted-for facilities in the cost matrix computation should reduce the travel cost from/to that origin/destination.

The spatial interaction model parameters are fixed to $\epsilon=1$, $\kappa = 1.025$, $\delta=0.0128$ and $\sigma_d=3\%\times \log(13)$  in the same fashion as \citep{ellam2018} to ensure that our results are comparable. The joint MCMC samplers \ref{alg:tractable_table_sampling_algorithm} and \ref{alg:intractable_sampling_algorithm} are run for $N=10^5$ steps and their samples are utilised to generate Figures \ref{fig:lower_dim_embedding_samples}, \ref{fig:posterior_mean_x_predictions} and Table \ref{tab:metrics}. In both samplers, $\mathbf{x}$-acceptance is at least $90\%$, ensuring numerical stability of the leapfrog integrator in HMC, and $\boldsymbol{\theta}$-acceptance ranges from $30\%$ to $70\%$ depending on the size of $\mathcal{C}$ and the corresponding concentration of measure. In the high-noise regime, importance sampling estimates for the normalising constant in \eqref{eq:boltzmann_gibbs_measure_normalising_constant} comprises of $100$ particles and a uniform temperature schedule of $50$ inverse temperatures. We ensure that at least $75\%$ of the signs are positive in order to obtain sufficiently representative samples of the normalising constant. The first $10^4$ $\mathbf{T},\mathbf{x},\boldsymbol{\theta}$ samples are discarded and a thinning of $90$ is applied to ensure that $10^3$ independent samples from the joint posterior are obtained. We note that in the case of Algorithm \ref{alg:tractable_table_sampling_algorithm} independent samples of $\mathbf{T}$ are directly accessible by virtue of closed-form sampling. In all intractable table sampling experiments, MB-Gibbs is employed to explore the table posterior marginal, due to its superior efficiency relative to MB-MH. The conditional expectations $\mathbb{E}[\mathbf{x}\vert \mathbf{y},\mathcal{C}]$ and $\mathbb{E}[\boldsymbol{\theta}\vert \mathbf{y},\mathcal{C}]$ are computed using the unbiased estimator in \citep{ellam2018}:
\begin{equation*}
    \mathbb{E}[\phi(\mathbf{x},\boldsymbol{\theta})\vert \mathbf{y},\mathcal{C}] = \frac{\sum_{n=1}^N \text{sign}(\boldsymbol{\theta}^{(n)}) \phi(\mathbf{x}^{(n)},\boldsymbol{\theta}^{(n)})}{\sum_{n=1}^N \text{sign}(\boldsymbol{\theta}^{(n)})},
\end{equation*}
where $\text{sign}(\boldsymbol{\theta}^{(n)})$ is the sign of sampled $\boldsymbol{\theta}^{(n)}$. 

The MCMC samples of the non-joint scheme in \citep{ellam2018} are retrieved by using the same specification as above. In order to draw a comparison between our method and the Neural network approach in \citep{gaskin2023}, we initialise $10^2$ random seeds and run the neural net for $10^3$ epochs giving a total of $10^5$ pseudo-samples. Similarly to the setup above, we apply the same burn-in and thinning to these pseudo-samples. In the free/variable noise regime, we learn the SDE noise together with $\mathbf{x}$ and $\boldsymbol{\theta}$. The SIM parameters are conditioned as depicted in Table \ref{fig:plate_diagram} using the same values as above.

In Figure \ref{fig:lower_dim_embedding_samples}, table and intensity samples, as well as the ground truth ODM, are jointly projected to 2D using a T-distributed stochastic neighbour embedding \citep{hinton2002} based on the $l_1$ metric and a perplexity equal to $30$. The validation statistics in Table \ref{tab:metrics} are defined as follows. The standardised root mean square error \citep{oshan2016} is equal to 
\begin{equation*}
    \text{SRMSE}(\mathbf{T},\mathbf{T}^{\mathcal{D}}) = \sqrt{\frac{\sum_{x\in\mathcal{X}} \left(T(x)-T^{\mathcal{D}}(x)\right)^2}{\vert \mathcal{X} \vert }} \left(\frac{\sum_{x\in\mathcal{X}} T(x)}{\vert \mathcal{X} \vert}\right)^{-1},
\end{equation*}
and the Sorensen similarity index is equal to
\begin{equation*}
    \text{SSI}(\mathbf{T},\mathbf{T}^{\mathcal{D}}) = \frac{1}{\vert \mathcal{X} \vert} \sum_{x\in\mathcal{X}} \frac{2\min\left(T(x),T^{\mathcal{D}}(x)\right)}{T(x)+T^{\mathcal{D}}(x)},
\end{equation*}
whose domain is $[0,1]$ and values closer to $1$ imply a better ODM fit. Moreover, we define the Markov Basis distance as
\begin{equation*}
    \text{MBD}(\mathbf{T},\mathbf{T}^{\mathcal{D}}) = 1/2 \sum_{x\in\mathcal{X}} \vert T(x) - T^{\mathcal{D}}(x) \vert.
\end{equation*}
By construction, this distance function gives an upper bound to the number of Markov basis moves of step size $\eta=1$ required for $\mathbf{T}$ to reach $\mathbf{T}^{\mathcal{D}}$ while staying in $\mathcal{T}^{\mathcal{D}}$.

Finally, the coverage probability of the ground truth table $\mathbf{T}^{\mathcal{D}}$ is calculated by first identifying the lower $L_q\left(\mathbf{T}^{(1:N)}(x)\right)$ and upper $U_q\left(\mathbf{T}^{(1:N)}(x)\right)$ boundaries of the high probability region containing $q\%$ of the total mass for each table cell $x \in \mathcal{X}$. Then, the $q\%$ cell coverage probability is equal to
\begin{equation*}
    \text{CP}_q(\mathbf{T},\mathbf{T}^{\mathcal{D}}) = 1/ \vert \mathcal{X} \vert \sum_{x\in\mathcal{X}} \mathds{1}\biggl\{L_q\left(\mathbf{T}^{(1:N)}(x)\right) \leq T^{\mathcal{D}}(x) \leq U_q\left(\mathbf{T}^{(1:N)}(x)\right)\biggl\},
\end{equation*}
where where $\mathds{1}$ is the indicator function, and $\mathbf{T}^{(1:N)}$ is the tensor of all table samples.

\section{Auxiliary results for the Cambridge application}\label{app:more_experiments}

In this section, we append additional experimental results of the large-scale application to Cambridge commuting patterns. Figure \ref{fig:posterior_table_intensity_mean_predictions} compares the smallest ODM reconstruction errors from each method in Table \ref{tab:contributions}. The doubly and 20\% cell constrained table significantly improves ground truth coverage of free cells (row 2) over the singly constrained intensity of \citep{ellam2018} (row 4) even though both use the same cost matrix. Therefore, augmenting $\mathcal{C}_T$ shrinks the support of free cells and diminishes their reconstruction error. Compared to the singly constrained intensity of \citep{gaskin2023}, the table posterior mean achieves a reduced SRMSE ($0.51$ versus $0.61$) and does not overestimate the trips to popular destinations. Such destinations include $7$ and $13$, which correspond to areas where the city's university premises and hospital are located and are highly attractive. The neural network's low noise predictions (row 3) are characterised by stronger attraction effects, which translates to over-estimation of trips to the aforementioned destinations. The higher noise regime has little effect on the neural network's trip prediction error as opposed to the table posterior's error, which is substantially reduced in the high noise regime. This suggests that continuous intensity is more susceptible to ODM overfitting than the discrete table. A diffuse SDE smooths the marginal distributions of free table cells without changing their support and allows for more contributions of the $\mathbb{P}\left(\mathbf{T}\vert \mathbf{x}, \boldsymbol{\theta}, \mathcal{C}\right)$ term to the posterior mean. 

\begin{figure}[!ht]
    \centering
    \input{./tex_figures/mean_table_heatmaps/trip_colorbar.tikz}
    \input{./tex_figures/mean_table_heatmaps/destination_demand_colorbar.tikz}
    \input{./tex_figures/mean_table_heatmaps/trip_error_colorbar.tikz}
    \input{./tex_figures/mean_table_heatmaps/destination_demand_error_colorbar.tikz}
\input{./tex_figures/mean_table_heatmaps/ground_truth_table.tikz}
\input{./tex_figures/mean_table_heatmaps/table_error_mean_heatmap_JointTableSIMLatentMCMC_high_noise_both_margins_permuted_cells_20percent_thinning1_burnin10000.tikz}    \input{./tex_figures/mean_table_heatmaps/intensity_error_mean_heatmap_NeuralABM_low_noise_row_margin_thinning1_burnin10000.tikz}
\input{./tex_figures/mean_table_heatmaps/intensity_error_mean_heatmap_SIMLatentMCMC_high_noise_row_margin_thinning1_burnin10000.tikz}
  \caption{Ground truth table $\mathbf{T}^{\mathcal{D}}$ (row 1), $\mathbf{T}^{\mathcal{D}}-\mathbb{E}\left[\mathbf{T}\vert\mathbf{y},\{\mathbf{T}_{\cdot+},\mathbf{T}_{+,\cdot},\mathbf{T}_{\mathcal{X}_2},\Lambda_{++}\}\right]$ (row 2) and $\mathbf{T}^{\mathcal{D}}-\mathbb{E}\left[\boldsymbol{\Lambda}\vert\mathbf{y},\{\boldsymbol{\Lambda}_{\cdot+}\}\right]$ (rows 3 \& 4) using Algorithm \ref{alg:intractable_sampling_algorithm} (row 2), the low noise Neural network in \citep{gaskin2023} (row 3) and the high noise MCMC scheme in \citep{ellam2018} (row 4). Each table's rows resemble the number of trips by destination, and vice versa. Cells with a black boundary correspond to fixed cells while $\checkmark$ cells cover the ground truth in the 99\% HPM. The high noise table predictions interpolate missing cell data in the vicinity of fixed cells and produce a more accurate trip ODM. The error profile similarities across all methods are attributed to the use of the same cost matrix, which captures the spatial covariance of trips.}%
    \label{fig:posterior_table_intensity_mean_predictions}
\end{figure}

\section{Posterior mean convergence consistency}

We obtain approximate results on the consistency between the discrete table and continuous intensity posterior mean estimators in Figure \ref{fig:convergence_diagnostic}. A ground truth approximation of Fisher's non-central hypergeometric's mean $\bar{\mathbf{T}}$ is leveraged to ensure that the following identity holds empirically:
\begin{equation*}
    \mathbb{E}_{T,\Lambda} \left[ \mathbf{T} \vert \mathbf{y}, \mathcal{C} \right] = \mathbb{E}_{\Lambda} \Bigr[ \underbrace{\mathbb{E}_{T} \left[ \mathbf{T} \vert \boldsymbol{\Lambda}, \mathcal{C} \right]}_{\bar{\mathbf{T}}} \Bigr].
\end{equation*}
A pattern of faster convergence rates in the limit of more data $\mathcal{C}_T$ is depicted, although hardly visible.

\begin{figure}[ht!]
    \centering
\input{./tex_figures/intensity_and_table_convergence/mean_table_and_intensity_relative_l_1_norm_vs_mcmc_iteration.tex}
    \captionsetup{width=.9\linewidth}
    \caption{Convergence of the $l_1$ norm of $\mathbb{E}[\mathbf{T}\vert\mathbf{y},\mathcal{C}] - \mathbb{E}[\boldsymbol{\Lambda}\vert\mathbf{y},\mathcal{C}]$ for different constraint sets $\mathcal{C}$ in the low (~\ref{subplot:doubly_constrained_low_noise_convergence},~\ref{subplot:doubly_10percent_cell_constrained_low_noise_convergence},~\ref{subplot:doubly_20percent_cell_constrained_low_noise_convergence}) and high (~\ref{subplot:doubly_constrained_high_noise_convergence},~\ref{subplot:doubly_10percent_cell_constrained_high_noise_convergence},~\ref{subplot:doubly_20percent_cell_constrained_high_noise_convergence}) noise regimes.}%
    \label{fig:convergence_diagnostic}
\end{figure}

\end{document}

%% file: my_macros.tex
\usepgfplotslibrary{colorbrewer}
\usetikzlibrary{patterns, patterns.meta, shapes.arrows, matrix, shapes.geometric}
\pgfplotsset{compat=newest}

\graphicspath{{../../figures/}}

\AddToHook{env/tabular/begin}{\let\input\@@input}

\input{insbox}

\tcbset{lemmastyle/.style={title={},breakable,colback=white,colframe=white,boxrule=0pt,boxsep=0pt,title after break={Lemma \thelemma\ (continued)}}}

\tcbset{theostyle/.style={title={},breakable,colback=white,colframe=white,boxrule=0pt,boxsep=0pt,title after break={Theorem \thetheorem\ (continued)}}}

\tcbset{corrstyle/.style={title={},breakable,colback=white,colframe=white,boxrule=0pt,boxsep=0pt,title after break={Corollary \thecorollary\ (continued)}}}

\tcbset{remarkstyle/.style={title={},breakable,colback=white,colframe=white,boxrule=0pt,boxsep=0pt,title after break={Remark \theremark\ (continued)}},left=0em,right=0em}

\tcbset{propositionstyle/.style={title={},breakable,colback=white,colframe=white,boxrule=0pt,boxsep=0pt,title after break={Proposition \theproposition\ (continued)}}}

\tcbset{proofstyle/.style={title={},breakable,colback=white,colframe=white,boxrule=0pt,boxsep=0pt,title after break={Proof \ (continued)}}}

\tcolorboxenvironment{lemma}{lemmastyle}
\tcolorboxenvironment{theorem}{theostyle}
\tcolorboxenvironment{corollary}{corrstyle}
\tcolorboxenvironment{proof}{proofstyle}
\tcolorboxenvironment{proposition}{propositionstyle}

\newcommand\BibTeX{{\rmfamily B\kern-.05em \textsc{i\kern-.025em b}\kern-.08em
T\kern-.1667em\lower.7ex\hbox{E}\kern-.125emX}}
\newcommand\numberthis{\addtocounter{equation}{1}\tag{\theequation}}

\pgfdeclareplotmark{mystar}{
    \node[star,opacity=1.0,star point ratio=2.25,minimum size=6pt,inner sep=0.0pt,draw=black,solid,fill=red] {};
}

\DeclareMathAlphabet\mathbfcal{OMS}{cmsy}{b}{n}

\definecolor{darkgray176}{RGB}{176,176,176}
\definecolor{colorblind_red}{RGB}{216, 27, 96}
\definecolor{colorblind_blue}{RGB}{30, 136, 229}
\definecolor{colorblind_yellow}{RGB}{255, 193, 7}
\definecolor{colorblind_orange}{RGB}{254, 97, 0}
\definecolor{colorblind_green}{RGB}{0, 77, 64}
\definecolor{cdarkblue}{RGB}{0,53,102}
\definecolor{corange}{RGB}{230,85,13}
\definecolor{cgrassgreen}{RGB}{0,128,0}
\definecolor{cskyblue}{RGB}{0,155,219}
\definecolor{cyellow}{RGB}{240,228,66}
\definecolor{cdarkgreen}{RGB}{0,100,0}
\definecolor{cpurple}{RGB}{145,30,180}
\definecolor{clightgreen}{RGB}{153,180,51}
\definecolor{cdarkred}{RGB}{153,0,0}
\definecolor{darkorange2551490}{RGB}{255,149,0}
\definecolor{limegreen018569}{RGB}{0,185,69}
\definecolor{orangered255440}{RGB}{255,44,0}
\definecolor{teal1293165}{RGB}{12,93,165}
\definecolor{slategray13291151}{RGB}{132,91,151}
\definecolor{green68178118}{RGB}{68,178,118}

\newlength{\flexcheckerboardsize}
\newcommand{\defineflexcheckerboard}[4]{
    \setlength{\flexcheckerboardsize}{#2}
    \pgfdeclarepatterninherentlycolored{#1}
        {\pgfpointorigin}{\pgfqpoint{2\flexcheckerboardsize}{2\flexcheckerboardsize}}
        {\pgfqpoint{2\flexcheckerboardsize}{2\flexcheckerboardsize}}%
        {
           \pgfsetfillcolor{#4}
           \pgfpathrectangle{\pgfpointorigin}{\pgfqpoint{2.1\flexcheckerboardsize}{2.1\flexcheckerboardsize}}
          \pgfusepath{fill}
          \pgfsetfillcolor{#3}
          \pgfpathrectangle{\pgfpointorigin}{\pgfqpoint{\flexcheckerboardsize}{\flexcheckerboardsize}}
          \pgfpathrectangle{\pgfqpoint{\flexcheckerboardsize}{\flexcheckerboardsize}}
            {\pgfqpoint{\flexcheckerboardsize}{\flexcheckerboardsize}}
            \pgfusepath{fill}
        }
}

\defineflexcheckerboard{flexcheckerboard_cdarkblue}{.3mm}{white}{cdarkblue}
\defineflexcheckerboard{flexcheckerboard_cskyblue}{.3mm}{white}{cskyblue}
\defineflexcheckerboard{flexcheckerboard_corange}{.3mm}{white}{corange}
\defineflexcheckerboard{flexcheckerboard_cgrassgreen}{.3mm}{white}{cgrassgreen}
\defineflexcheckerboard{flexcheckerboard_cdarkred}{.3mm}{white}{cdarkred}
\defineflexcheckerboard{flexcheckerboard_cpurple}{.3mm}{white}{cpurple}
\defineflexcheckerboard{flexcheckerboard_cyellow}{.3mm}{white}{cyellow}
\defineflexcheckerboard{flexcheckerboard_black}{.3mm}{white}{black}

\def\mydrawopacity{0.1}
\def\myfillopacity{0.4}

\newlength{\flexlinelength}
\newlength{\flexlinewidth}
\newcommand{\definehorrizontallines}[5]{
    \setlength{\flexlinelength}{#2}
    \setlength{\flexlinewidth}{#3}
    \pgfdeclarepatterninherentlycolored{#1}
        {\pgfpointorigin}{\pgfpoint{2*\flexlinewidth}{100pt}}
        {\pgfpoint{2*\flexlinewidth}{100pt}}%
        {
        \pgfsetfillcolor{#4}
        \pgfpathrectangle{\pgfpointorigin}{\pgfpoint{\flexlinelength+\flexlinewidth}{100pt}}
        \pgfusepath{fill}
        \pgfsetfillcolor{#5}
        \pgfpathrectangle{\pgfpoint{\flexlinelength}{0pt}}{\pgfpoint{\flexlinelength+\flexlinewidth}{100pt}}
        \pgfusepath{fill}
    }
}

\definehorrizontallines{horrizontallines_cskyblue}{0.5pt}{0.5pt}{white}{cskyblue}
\definehorrizontallines{horrizontallines_black}{0.5pt}{0.5pt}{white}{black}

\pgfplotsset{
    colormap={redwhiteblue}{
        rgb255(0.0)=(59,76,192);
        rgb255(0.5)=(255,255,255);
        rgb255(1.0)=(180,4,38);
    },
    colormap={yellowwhitepurple}{
        rgb255(0.0)=(68,1,84);
        rgb255(0.14)=(70,50,127);
        rgb255(0.29)=(54,92,141);
        rgb255(0.43)=(39,127,142);
        rgb255(0.5)=(255,255,255);
        rgb255(0.57)=(31,161,135);
        rgb255(0.71)=(74,194,109);
        rgb255(0.86)=(159,218,58);
        rgb255(1.0)=(253,231,37);
    },
    colormap={bluegreen}{
        rgb255(0.0)=(24,8,163);
        rgb255(0.2)=(7,69,142);
        rgb255(0.4)=(5,105,105);
        rgb255(0.6)=(7,137,66);
        rgb255(0.8)=(15,168,8);
        rgb255(1.0)=(97,193,9);
    },
    colormap={yellowpurple}{
        rgb255(0.0)=(13,8,135);
        rgb255(0.14)=(84,2,163);
        rgb255(0.29)=(139,10,165);
        rgb255(0.43)=(185,50,137);
        rgb255(0.57)=(219,92,104);
        rgb255(0.71)=(244,136,73);
        rgb255(0.86)=(254,188,43);
        rgb255(1.0)=(240,249,33);
    },
    colormap={yellowblue}{
        rgb255(0.0)=(255,165,0);
        rgb255(0.1)=(255,255,0);
        rgb255(0.4)=(255,255,224);
        rgb255(0.5)=(255,255,255);
        rgb255(0.6)=(173,216,230);
        rgb255(0.9)=(0,0,255);
        rgb255(1.0)=(0,0,139);
    },
    colormap={redgreen}{
        rgb255(0.0)=(139,0,0);
        rgb255(0.1)=(255,0,0);
        rgb255(0.3)=(240,128,128);
        rgb255(0.5)=(255,255,255);
        rgb255(0.6)=(152,251,152);
        rgb255(0.9)=(0,128,0);
        rgb255(1.0)=(0,100,0);
    },
}

%% file: tex_tables/summaries/contributions.tex
  \hline
   $\mathcal{C}$ &
    Constrained ODM &
    This work &
    \citepalias{ellam2018} &
    \citepalias{gaskin2023} &
    $\mathbb{P}(\mathbf{\textcolor[HTML]{1E88E5}{T}}\vert \boldsymbol{\textcolor[HTML]{FFC20A}{\Lambda}}, \mathcal{C})$ \\ \hline
  $\bigl\{\textcolor[HTML]{FFC20A}{\Lambda}_{++}\bigl\}$ &
    Totally &
    $\checkmark$ &
    $\checkmark$ &
    $\checkmark$ &
    - \\ \hline
  $\bigl\{\boldsymbol{\textcolor[HTML]{FFC20A}{\Lambda}}_{\cdot +}\bigl\}$ &
    Singly &
    $\checkmark$ &
    $\checkmark$ &
    $\checkmark$ &
    - \\ \hline
  $\bigl\{\textcolor[HTML]{1E88E5}{T}_{++},\textcolor[HTML]{FFC20A}{\Lambda}_{++}\bigl\}$ &
    Totally &
    $\checkmark$ &
    $\times$ &
    $\times$ &
    Multinomial $\; \left(\textcolor[HTML]{1E88E5}{T}_{++},\textcolor[HTML]{FFC20A}{\Lambda}_{++}\right)$ \\ \hline
  $\bigl\{\mathbf{\textcolor[HTML]{1E88E5}{T}}_{\cdot +},\boldsymbol{\textcolor[HTML]{FFC20A}{\Lambda}}_{\cdot +}\bigl\}$ or $\bigl\{\mathbf{\textcolor[HTML]{1E88E5}{T}}_{\cdot +},\textcolor[HTML]{FFC20A}{\Lambda}_{++}\bigl\}$ &
    Singly &
    $\checkmark$ &
    $\times$ &
    $\times$ &
    Product Multinomial $\; \left(\mathbf{\textcolor[HTML]{1E88E5}{T}}_{\cdot +},\frac{\boldsymbol{\textcolor[HTML]{FFC20A}{\Lambda}}_{\cdot +}}{\textcolor[HTML]{FFC20A}{\Lambda}_{++}}\right)$ \\ \hline
  $\bigl\{\mathbf{\textcolor[HTML]{1E88E5}{T}}_{\cdot +},\mathbf{\textcolor[HTML]{1E88E5}{T}}_{+ \cdot},\textcolor[HTML]{FFC20A}{\Lambda}_{++}\bigl\}$ &
    Doubly &
    $\checkmark$ &
    $\times$ &
    $\times$ &
    Fisher's non-central hypergeometric $\; \left(\mathbf{\textcolor[HTML]{1E88E5}{T}}_{\cdot +},\mathbf{\textcolor[HTML]{1E88E5}{T}}_{+\cdot},\frac{\textcolor[HTML]{FFC20A}{\Lambda}_{++} \boldsymbol{\textcolor[HTML]{FFC20A}{\Lambda}}_{\cdot\cdot}}{\boldsymbol{\textcolor[HTML]{FFC20A}{\Lambda}}_{\cdot +}\boldsymbol{\textcolor[HTML]{FFC20A}{\Lambda}}_{+\cdot}}\right)$ \\ \hline
  $\bigl\{\mathbf{\textcolor[HTML]{1E88E5}{T}}_{\cdot +}, \mathbf{\textcolor[HTML]{1E88E5}{T}}_{+\cdot}, \newline  \mathbf{\textcolor[HTML]{1E88E5}{T}}_{\mathcal{X}'}, \textcolor[HTML]{FFC20A}{\Lambda}_{++}, \newline        \forall \; \mathcal{X}'\subseteq\mathcal{P}(\mathcal{X})\bigl\}$ &
    Doubly and cell &
    $\checkmark$ &
    $\times$ &
    $\times$ &
    Constrained Fisher's non-central hypergeometric $\; \left(\mathbf{\textcolor[HTML]{1E88E5}{T}_{\cdot +}},\mathbf{\textcolor[HTML]{1E88E5}{T}}_{+\cdot},\frac{\textcolor[HTML]{FFC20A}{\Lambda}_{++} \boldsymbol{\textcolor[HTML]{FFC20A}{\Lambda}}_{\cdot\cdot}}{\boldsymbol{\textcolor[HTML]{FFC20A}{\Lambda}}_{\cdot +}\boldsymbol{\textcolor[HTML]{FFC20A}{\Lambda}}_{+\cdot}}\right)$ \\ \hline

%% file: tex_figures/intensity_and_table_convergence/mean_table_relative_l_1_norm_vs_iteration_comparison_by_table_dim__table_total_K_1000.tex
\begin{tikzpicture}

  \begin{axis}[
  legend style={fill opacity=0.8, draw opacity=1, text opacity=1, draw=none},
  every axis plot/.append style={ultra thick},
  tick pos=both,
  x grid style={darkgray176},
  xlabel={MCMC Iteration},
  restrict x to domain=0:1000,
  xtick distance={200},
  y grid style={darkgray176},
  ylabel={$\bigg\| {\mathbb{{E}}[\mathbb{{E}}[\mathbf{T}\vert\mathbf{T}_{\cdot +},\boldsymbol{\Lambda}]]}  - \mathbf{T}_{\cdot +}\boldsymbol{\Lambda}/\boldsymbol{\Lambda}_{\cdot +} \bigg\|_1^{2} / \bigg\| \mathbf{T}_{\cdot +}\boldsymbol{\Lambda}/\boldsymbol{\Lambda}_{\cdot +} \bigg\|_1^{2}$},
  restrict y to domain=0.0:0.05,
  ytick distance={0.01},
  unbounded coords=jump
  ]

    \addplot [colorblind_red]
    table [each nth point={1}] {./tex_tables/intensity_and_table_convergence/dim_2x3_total_100_experiment_title_row_margin.dat};
    \label{subplot:dim_2x3_total_100_experiment_title_row_margin}
  \addplot [colorblind_blue] 
    table [each nth point={1}] {./tex_tables/intensity_and_table_convergence/dim_2x3_total_5000_experiment_title_row_margin.dat};
    \label{subplot:dim_2x3_total_5000_experiment_title_row_margin}
  \addplot [colorblind_yellow]
    table [each nth point={1}] {./tex_tables/intensity_and_table_convergence/dim_33x33_total_5000_experiment_title_row_margin.dat};
    \label{subplot:dim_33x33_total_5000_experiment_title_row_margin}
  \addplot [colorblind_green]
    table [each nth point={1}] {./tex_tables/intensity_and_table_convergence/dim_33x33_total_100_experiment_title_row_margin.dat};
    \label{subplot:dim_33x33_total_100_experiment_title_row_margin}

\end{axis}

\matrix [
  draw,
  matrix of nodes,
  column sep=0.0,
  fill opacity=0.0,
  draw opacity=0.0,
  text opacity=1.0,
  anchor=north west,
  node font=\small,
  column 1/.style={anchor=west}
] at (2.5,5.6) {
  \ref*{subplot:dim_2x3_total_100_experiment_title_row_margin} $\dim({\mathbf{T}})=(2,3), T_{++} = 100$ \\  
  \ref*{subplot:dim_2x3_total_5000_experiment_title_row_margin} $\dim({\mathbf{T}})=(2,3), T_{++} = 5000$\\
  \ref*{subplot:dim_33x33_total_5000_experiment_title_row_margin} $\dim({\mathbf{T}})=(33,33), T_{++} = 5000$ \\
  \ref*{subplot:dim_33x33_total_100_experiment_title_row_margin} $\dim({\mathbf{T}})=(33,33), T_{++} = 100$ \\
};

\end{tikzpicture}

%% file: tex_figures/intensity_and_table_convergence/mean_table_relative_l_1_norm_vs_iteration_comparison_by_proposal_K_1000.tex
\begin{tikzpicture}
  
  \begin{axis}[
  legend cell align={left},
  legend style={fill opacity=0.8, draw opacity=1, text opacity=1, draw=none},
  every axis plot/.append style={ultra thick},
  tick pos=both,
  xlabel={MCMC Iteration},
  restrict x to domain=0:10000,
  xtick={0,1000,2000,...,9000,10000},
  ylabel={$\bigg\| {\mathbb{{E}}[\mathbb{{E}}[\mathbf{T}|\mathcal{C}_T,\boldsymbol{\Lambda}]]}  - g(\boldsymbol{\Lambda}) \bigg\|_1^{2}$},
  restrict y to domain=0:0.5,
  ytick distance={0.1},
  ]

  \addplot [colorblind_blue]
  table [each nth point={1},filter discard warning={false}, unbounded coords={discard}]{./tex_tables/intensity_and_table_convergence/proposal_direct_sampling_experiment_title_row_margin.dat};
  \label{subplot:proposal_direct_sampling_experiment_title_row_margin}
  \addplot [colorblind_green]
  table [each nth point={1},filter discard warning={false}, unbounded coords={discard}]{./tex_tables/intensity_and_table_convergence/proposal_degree_one_experiment_title_row_margin.dat};
  \label{subplot:proposal_degree_one_experiment_title_row_margin}
  \addplot [colorblind_green,dashed] 
    table [each nth point={1},filter discard warning={false}, unbounded coords={discard}] {./tex_tables/intensity_and_table_convergence/proposal_degree_one_experiment_title_both_margins.dat};
  \label{subplot:proposal_degree_one_experiment_title_both_margins}
  \addplot [colorblind_red] 
    table [each nth point={1},filter discard warning={false}, unbounded coords={discard}] {./tex_tables/intensity_and_table_convergence/proposal_degree_higher_experiment_title_row_margin.dat};
  \label{subplot:proposal_degree_higher_experiment_title_row_margin}
  \addplot [colorblind_red,dashed] 
    table [each nth point={1},filter discard warning={false}, unbounded coords={discard}] {./tex_tables/intensity_and_table_convergence/proposal_degree_higher_experiment_title_both_margins.dat};
  \label{subplot:proposal_degree_higher_experiment_title_both_margins}

  \addlegendimage{draw=black, fill=black, mark=none} 
  \label{subplot:row_margin}
  \addlegendimage{draw=black, fill=black, mark=none, dashed} 
  \label{subplot:both_margins}

  \end{axis}

  \matrix [
        draw,
        matrix of nodes,
        column sep=0.0,
        fill opacity=0.0,
        draw opacity=0.0,
        text opacity=1.0,
        anchor=north west,
        node font=\small,
        column 1/.style={anchor=west}
    ] at (3.655,5.6) {
        \ref*{subplot:row_margin} $\mathcal{C}_T=\left\{\mathbf{T}_{\cdot+}\right\}$ \\  
        \ref*{subplot:both_margins} $\mathcal{C}_T=\left\{\mathbf{T}_{\cdot+},\mathbf{T}_{+\cdot}\right\}$ \\  
        \ref*{subplot:proposal_direct_sampling_experiment_title_row_margin} Direct sampling \\
        \ref*{subplot:proposal_degree_one_experiment_title_row_margin},\ref*{subplot:proposal_degree_one_experiment_title_both_margins}  Metropolis Hastings \\
        \ref*{subplot:proposal_degree_higher_experiment_title_row_margin},\ref*{subplot:proposal_degree_higher_experiment_title_both_margins}  Gibbs \\
    };

\end{tikzpicture}
  

%% file: tex_figures/2d_sample_projections/table_2d_isomap_label_by_experiment_title_gamma_burnin10000_thinning50_euclidean_distance_nearest_neighbours100.tikz
\begin{tikzpicture}[
    hatch distance/.store in=\hatchdistance,
    hatch distance=10pt,
    hatch thickness/.store in=\hatchthickness,
    hatch thickness=2pt
]

\pgfplotsset{
    table/search path={./tex_tables/2d_sample_projections/},
}

\begin{axis}[
%
xtick pos=bottom,
restrict x to domain=-100:120,
xmin=-100,xmax=120,
x grid style={darkgray176},
xlabel={Projected dimension 1},
%
ytick pos=left,
restrict y to domain=-100:150,
ymin=-100,ymax=150,
y grid style={darkgray176},
ylabel={Projected dimension 2},
]






\addplot[
    only marks,
    mark=triangle*,
    pattern=flexcheckerboard_cdarkblue,
    draw=cdarkblue,
    draw opacity=\mydrawopacity,
    fill opacity=\myfillopacity,
] table {table_JointTableSIMLatentMCMC_row_margin_low_noise.dat};
\label{subplot:table_row_margin_low_noise}

\addplot [
    only marks,
    mark=triangle*,
    draw=cdarkblue,
    fill=cdarkblue,
    draw opacity=\mydrawopacity,
    fill opacity=\myfillopacity,
]
table {table_JointTableSIMLatentMCMC_row_margin_low_noise.dat};
\label{subplot:table_row_margin_high_noise}


\addplot[
    only marks,
    mark=square*,
    pattern=flexcheckerboard_cgrassgreen,
    draw=cgrassgreen,
    draw opacity=\mydrawopacity,
    fill opacity=\myfillopacity,
] table {table_JointTableSIMLatentMCMC_both_margins_low_noise.dat};
\label{subplot:table_both_margins_low_noise}

\addplot [
    only marks,
    mark=square*,
    draw=cgrassgreen,
    fill=cgrassgreen,
    draw opacity=\mydrawopacity,
    fill opacity=\myfillopacity,
]
table {table_JointTableSIMLatentMCMC_both_margins_high_noise.dat};
\label{subplot:table_both_margins_high_noise}


\addplot[
    only marks,
    mark=pentagon*,
    pattern=flexcheckerboard_corange,
    draw=corange,
    draw opacity=\mydrawopacity,
    fill opacity=\myfillopacity,
] table {table_JointTableSIMLatentMCMC_both_margins_permuted_cells_10percent_low_noise.dat};
\label{subplot:table_both_margins_permuted_cells_10percent_low_noise}

\addplot [
    only marks,
    mark=pentagon*,
    draw=corange,
    fill=corange,
    draw opacity=\mydrawopacity,
    fill opacity=\myfillopacity,
]
table {table_JointTableSIMLatentMCMC_both_margins_permuted_cells_10percent_high_noise.dat};
\label{subplot:table_both_margins_permuted_cells_10percent_high_noise}


\addplot[
    only marks,
    mark=diamond*,
    pattern=flexcheckerboard_cpurple,
    draw=cpurple,
    draw opacity=\mydrawopacity,
    fill opacity=\myfillopacity,
] table {table_JointTableSIMLatentMCMC_both_margins_permuted_cells_20percent_low_noise.dat};
\label{subplot:table_both_margins_permuted_cells_20percent_low_noise}

\addplot [
    only marks,
    mark=diamond*,
    draw=cpurple,
    fill=cpurple,
    draw opacity=\mydrawopacity,
    fill opacity=\myfillopacity,
]
table {table_JointTableSIMLatentMCMC_both_margins_permuted_cells_20percent_high_noise.dat};
\label{subplot:table_both_margins_permuted_cells_20percent_high_noise}


\addplot [
    only marks,
    mark=mystar, 
    draw=cdarkred, 
    fill=cdarkred, 
]
table {ground_truth.dat};
\label{subplot:ground_truth_table}

\addlegendimage{
    only marks,
    draw=black, 
    fill=black,  
    pattern=flexcheckerboard_black,
    opacity=\mydrawopacity, 
    fill opacity=\myfillopacity,
} 
\label{subplot:table_low_noise}
\addlegendimage{
    only marks,
    draw=black, 
    fill=black,
    opacity=\mydrawopacity, 
    fill opacity=\myfillopacity,
} 
\label{subplot:table_high_noise}

\end{axis}

\matrix [
        draw,
        matrix of nodes,
        column sep=0.0,
        fill opacity=0.0,
        draw opacity=0.0,
        text opacity=1.0,
        anchor=north west,
        node font=\small,
        column 1/.style={anchor=west},
        column 2/.style={anchor=west}
    ] at (0.1,5.7) {
        \ref*{subplot:ground_truth_table} $\mathbf{T}^\mathcal{D}$ \ref*{subplot:table_low_noise} $\gamma = 10^4$ \ref*{subplot:table_high_noise} $\gamma = 10^2$ \\
        \ref*{subplot:table_row_margin_low_noise},\ref*{subplot:table_row_margin_high_noise} $\mathcal{C}_T=\left\{\mathbf{T}_{\cdot+}\right\}$ 
        \ref*{subplot:table_both_margins_low_noise},
        \ref*{subplot:table_both_margins_high_noise}
        $\mathcal{C}_T=\left\{\mathbf{T}_{\cdot+},\mathbf{T}_{+,\cdot}\right\}$ \\
        \ref*{subplot:table_both_margins_permuted_cells_10percent_low_noise},
        \ref*{subplot:table_both_margins_permuted_cells_10percent_high_noise}
        $\mathcal{C}_T=\left\{\mathbf{T}_{\cdot+},\mathbf{T}_{+,\cdot},\mathbf{T}_{\mathcal{X}_1}\right\}$
        \ref*{subplot:table_both_margins_permuted_cells_20percent_low_noise},
        \ref*{subplot:table_both_margins_permuted_cells_20percent_high_noise}:
        $\mathcal{C}_T=\left\{\mathbf{T}_{\cdot+},\mathbf{T}_{+,\cdot},\mathbf{T}_{\mathcal{X}_2}\right\}$ \\
    };

\end{tikzpicture}

%% file: tex_figures/2d_sample_projections/intensity_2d_isomap_label_by_experiment_title_gamma_burnin10000_thinning50_euclidean_distance_nearest_neighbours100.tikz
\begin{tikzpicture}

    \pgfplotsset{
        table/search path={./tex_tables/2d_sample_projections/},
    }
    
    \begin{axis}[
    %
    xtick pos=bottom,
    restrict x to domain=-100:120,
    xmin=-100,xmax=120,
    x grid style={darkgray176},
    xlabel={Projected dimension 1},
    %
    ytick pos=left,
    restrict y to domain=-100:180,
    ymin=-100,ymax=180,
    y grid style={darkgray176},
    ylabel={Projected dimension 2},
    ]
    

    \definecolor{cdarkblue}{RGB}{0,53,102}
    \definecolor{corange}{RGB}{230,85,13}
    \definecolor{cgrassgreen}{RGB}{0,128,0}
    \definecolor{cskyblue}{RGB}{0,155,219}
    \definecolor{cyellow}{RGB}{240,228,66}
    \definecolor{cdarkgreen}{RGB}{0,100,0}
    \definecolor{cpurple}{RGB}{145,30,180}
    \definecolor{clightgreen}{RGB}{153,180,51}
    \definecolor{cdarkred}{RGB}{153,0,0}
    

    \addplot [
        only marks,
        mark=triangle*, 
        pattern=flexcheckerboard_cyellow,
        draw=cyellow,
        opacity=\mydrawopacity, 
        fill opacity=\myfillopacity,
    ]
    table {intensity_SIMLatentMCMC_row_margin_low_noise.dat};
    \label{subplot:intensity_SIMLatentMCMC_row_margin_low_noise}
    
    \addplot [
        only marks,
        mark=triangle*, 
        fill=cyellow,
        draw=cyellow,
        opacity=\mydrawopacity, 
        fill opacity=\myfillopacity,
    ]
    table {intensity_SIMLatentMCMC_row_margin_high_noise.dat};
    \label{subplot:intensity_SIMLatentMCMC_row_margin_high_noise}
    

    \addplot [
        only marks,
        mark=triangle*, 
        pattern=flexcheckerboard_cskyblue,
        draw=cskyblue, 
        opacity=\mydrawopacity, 
        fill opacity=\myfillopacity,
    ]
    table {intensity_NeuralABM_row_margin_low_noise.dat};
    \label{subplot:intensity_NeuralABM_row_margin_low_noise}
    
    \addplot [
        only marks,
        mark=triangle*, 
        fill=cskyblue,
        draw=cskyblue, 
        opacity=\mydrawopacity, 
        fill opacity=\myfillopacity,
    ]
    table {intensity_NeuralABM_row_margin_high_noise.dat};
    \label{subplot:intensity_NeuralABM_row_margin_high_noise}

    \addplot [
        only marks,
        mark=triangle*, 
        pattern=horrizontallines_cskyblue,
        draw=cskyblue,
        opacity=\mydrawopacity, 
        fill opacity=\myfillopacity,
    ]
    table {intensity_NeuralABM_row_margin_learned_noise.dat};
    \label{subplot:intensity_NeuralABM_row_margin_learned_noise}
        

    \addplot [
        only marks,
        mark=triangle*,
        pattern=flexcheckerboard_cdarkblue,
        draw=cdarkblue,
        opacity=\mydrawopacity, 
        fill opacity=\myfillopacity,
    ]
    table {intensity_JointTableSIMLatentMCMC_row_margin_low_noise.dat};
    \label{subplot:intensity_JointTableSIMLatentMCMC_row_margin_low_noise}
    
    \addplot [
        only marks,
        mark=triangle*, 
        draw=cdarkblue, 
        fill=cdarkblue, 
        opacity=\mydrawopacity, 
        fill opacity=\myfillopacity,
    ]
    table {intensity_JointTableSIMLatentMCMC_row_margin_high_noise.dat};
    \label{subplot:intensity_JointTableSIMLatentMCMC_row_margin_high_noise}
    

    \addplot [
        only marks,
        mark=square*, 
        pattern=flexcheckerboard_cgrassgreen,
        draw=cgrassgreen, 
        opacity=\mydrawopacity, 
        fill opacity=\myfillopacity,
    ]
    table {intensity_JointTableSIMLatentMCMC_both_margins_low_noise.dat};
    \label{subplot:intensity_JointTableSIMLatentMCMC_both_margins_low_noise}
    
    \addplot [
        only marks,
        mark=square*, 
        fill=cgrassgreen, 
        draw=cgrassgreen, 
        opacity=\mydrawopacity, 
        fill opacity=\myfillopacity,
    ]
    table {intensity_JointTableSIMLatentMCMC_both_margins_high_noise.dat};
    \label{subplot:intensity_JointTableSIMLatentMCMC_both_margins_high_noise}

    
    \addplot [
        only marks,
        mark=pentagon*, 
        pattern=flexcheckerboard_corange,
        draw=corange,
        opacity=\mydrawopacity, 
        fill opacity=\myfillopacity,
    ]
    table {intensity_JointTableSIMLatentMCMC_both_margins_permuted_cells_10percent_low_noise.dat};
    \label{subplot:intensity_JointTableSIMLatentMCMC_both_margins_permuted_cells_10percent_low_noise}
    
    \addplot [
        only marks,
        mark=pentagon*, 
        fill=corange,
        draw=corange, 
        opacity=\mydrawopacity, 
        fill opacity=\myfillopacity,
    ]
    table {intensity_JointTableSIMLatentMCMC_both_margins_permuted_cells_10percent_high_noise.dat};
    \label{subplot:intensity_JointTableSIMLatentMCMC_both_margins_permuted_cells_10percent_high_noise}
    
    
    \addplot [
        only marks,
        mark=diamond*,
        pattern=flexcheckerboard_cpurple, 
        draw=cpurple, 
        opacity=\mydrawopacity, 
        fill opacity=\myfillopacity,
    ]
    table {intensity_JointTableSIMLatentMCMC_both_margins_permuted_cells_20percent_low_noise.dat};
    \label{subplot:intensity_JointTableSIMLatentMCMC_both_margins_permuted_cells_20percent_low_noise}
    
    \addplot [
        only marks,
        mark=diamond*, 
        fill=cpurple, 
        draw=cpurple, 
        opacity=\mydrawopacity, 
        fill opacity=\myfillopacity,
    ]
    table {intensity_JointTableSIMLatentMCMC_both_margins_permuted_cells_20percent_high_noise.dat};
    \label{subplot:intensity_JointTableSIMLatentMCMC_both_margins_permuted_cells_20percent_high_noise}

    \addplot [
        only marks,
        mark=mystar,
        fill=cdarkred, 
        draw=cdarkred,
    ]
    table {ground_truth.dat};
    \label{subplot:ground_truth_intensity}

    \addlegendimage{
        only marks,
        mark=*, 
        pattern=flexcheckerboard_black,
        draw=black,
        opacity=\mydrawopacity, 
        fill opacity=\myfillopacity,
    } 
    \label{subplot:intensity_low_noise}
    \addlegendimage{
        only marks,
        mark=*, 
        draw=black, 
        fill=black,
        opacity=\mydrawopacity, 
        fill opacity=\myfillopacity,
    } 
    \label{subplot:intensity_high_noise}
    \addlegendimage{
        only marks,
        mark=*,
        pattern=horrizontallines_black, 
        draw=black,
        opacity=\mydrawopacity, 
        fill opacity=\myfillopacity,
    } 
    \label{subplot:variable_noise}
    
    \end{axis}

        \matrix [
            draw,
            matrix of nodes,
            column sep=0.0,
            fill opacity=0.0,
            draw opacity=0.0,
            text opacity=1.0,
            anchor=north west,
            node font=\small,
            column 1/.style={anchor=west},
            column 2/.style={anchor=west}
        ] at  (0.1,5.7) {
            \ref*{subplot:ground_truth_table} $\mathbf{T}^\mathcal{D}$ \ref*{subplot:intensity_low_noise} $\gamma = 10^4$ \; \ref*{subplot:intensity_high_noise} $\gamma = 10^2$\; \ref*{subplot:variable_noise} $\gamma =\;$ free\;
            \ref*{subplot:intensity_JointTableSIMLatentMCMC_row_margin_low_noise},\ref*{subplot:intensity_JointTableSIMLatentMCMC_row_margin_high_noise} $\mathcal{C}_T=\left\{\mathbf{T}_{\cdot+}\right\}$ \\
            \ref*{subplot:intensity_SIMLatentMCMC_row_margin_low_noise},\ref*{subplot:intensity_SIMLatentMCMC_row_margin_high_noise} \citepalias{ellam2018} \ref*{subplot:intensity_NeuralABM_row_margin_low_noise},\ref*{subplot:intensity_NeuralABM_row_margin_high_noise},\ref*{subplot:intensity_NeuralABM_row_margin_learned_noise} \citepalias{gaskin2023} $\mathcal{C}_{\Lambda}=\left\{\boldsymbol{\Lambda}_{\cdot+}\right\}$ \;  \ref*{subplot:intensity_JointTableSIMLatentMCMC_both_margins_low_noise},\ref*{subplot:intensity_JointTableSIMLatentMCMC_both_margins_high_noise} $\mathcal{C}_T=\left\{\mathbf{T}_{\cdot+},\mathbf{T}_{+\cdot}\right\}$ \\
        };
        
        \matrix [
            draw,
            matrix of nodes,
            column sep=0.0,
            fill opacity=0.0,
            draw opacity=0.0,
            text opacity=1.0,
            anchor=north west,
            node font=\small,
            column 1/.style={anchor=west},
            column 2/.style={anchor=west}
        ] at  (0.1,4.9) {
            \ref*{subplot:intensity_JointTableSIMLatentMCMC_both_margins_permuted_cells_10percent_low_noise},
            \ref*{subplot:intensity_JointTableSIMLatentMCMC_both_margins_permuted_cells_10percent_high_noise}
            $\mathcal{C}_T=\left\{\mathbf{T}_{\cdot+}\;,\mathbf{T}_{+,\cdot}\;,\mathbf{T}_{\mathcal{X}_1}\right\}$ 
            \ref*{subplot:intensity_JointTableSIMLatentMCMC_both_margins_permuted_cells_20percent_low_noise},
            \ref*{subplot:intensity_JointTableSIMLatentMCMC_both_margins_permuted_cells_20percent_high_noise}:
            $\mathcal{C}_T=\left\{\mathbf{T}_{\cdot+}\;,\mathbf{T}_{+,\cdot}\;,\mathbf{T}_{\mathcal{X}_2}\right\}$ \\
        };

    \end{tikzpicture}
    

%% file: tex_tables/summaries/table_metrics.tex
%

  \hlineB{3.0}
  \textit{Constrained ODM} &
    $\mathcal{C}$ &
    \textit{Method} &
    $\gamma$ &
    \textit{Quantity} &
    $\vert \mathcal{M} \vert$ &
    \textit{$\mathbb{E} \left[\text{MDB}\left(\mathbf{T}^{(1:N)},\mathbf{T}^{\mathcal{D}}\right)\right]$} &
    \multicolumn{1}{l}{\textit{$\text{SSI}\left(\mathbb{E} \left[\mathbf{T}^{(1:N)}\right],\mathbf{T}^{\mathcal{D}}\right)$}} &
    \textit{$\text{SRMSE}\left(\mathbb{E} \left[\mathbf{T}^{(1:N)}\right],\mathbf{T}^{\mathcal{D}}\right)$} &
    \textit{$\text{CP}_{99}(\mathbf{T},\mathbf{T}^{\mathcal{D}})$} \\ \hlineB{3.0}
     &
     &
     &
     &
    $\boldsymbol{\textcolor[HTML]{FFC20A}{\Lambda}}\vert\mathcal{C},\mathbf{y}$ &
    - &
    - &
    $\mathbf{0.72}$ &
    $0.71$ &
    $0.22$ \\
   &
     &
     &
    \multirow{-2}{*}{$10^4$} &
    $\mathbf{\textcolor[HTML]{1E88E5}{T}}\vert\mathcal{C},\mathbf{y}$ &
    - &
    $9986$ &
    $0.54$ &
    $0.95$ &
    $0.70$ \\
   &
     &
     &
     &
    $\boldsymbol{\textcolor[HTML]{FFC20A}{\Lambda}}\vert\mathcal{C},\mathbf{y}$ &
    - &
    - &
    $0.48$ &
    $0.83$ &
    $0.37$ \\
   &
    \multirow{-4}{*}{$\textcolor[HTML]{FFC20A}{\Lambda}_{++}$} &
     &
    \multirow{-2}{*}{$10^2$} &
    $\mathbf{\textcolor[HTML]{1E88E5}{T}}\vert\mathcal{C},\mathbf{y}$ &
    - &
    $12684$ &
    $0.70$ &
    $0.85$ &
    $\mathbf{0.94}^{\mbox{\large $*$}}$ \\
   &
     &
     &
     &
    $\boldsymbol{\textcolor[HTML]{FFC20A}{\Lambda}}\vert\mathcal{C},\mathbf{y}$ &
    - &
    - &
    $\mathbf{0.72}$ &
    $0.73$ &
    $0.20$ \\
   &
     &
     &
    \multirow{-2}{*}{$10^4$} &
    $\mathbf{\textcolor[HTML]{1E88E5}{T}}\vert\mathcal{C},\mathbf{y}$ &
    $0$ &
    $\mathbf{7385}$ &
    $\mathbf{0.71}$ &
    $0.73$ &
    $0.65$ \\
   &
     &
     &
     &
    $\boldsymbol{\textcolor[HTML]{FFC20A}{\Lambda}}\vert\mathcal{C},\mathbf{y}$ &
    - &
    - &
    $0.70$ &
    $0.70$ &
    $0.41$ \\
   &
    \multirow{-4}{*}{$\textcolor[HTML]{1E88E5}{T}_{++},\textcolor[HTML]{FFC20A}{\Lambda}_{++}$} &
    \multirow{-8}{*}{This work} &
    \multirow{-2}{*}{$10^2$} &
    $\mathbf{\textcolor[HTML]{1E88E5}{T}}\vert\mathcal{C},\mathbf{y}$ &
    $0$ &
    $7600$ &
    $0.69$ &
    $0.70$ &
    $0.68$ \\ \Xcline{2-10}{0.01pt}
   &
     &
    \multicolumn{1}{l}{} &
    $10^4$ &
     &
    - &
    - &
    $\mathbf{0.72}$ &
    $0.74$ &
    $0.21$ \\
   &
    \multirow{-2}{*}{$\textcolor[HTML]{FFC20A}{\Lambda}_{++}$} &
    \multicolumn{1}{l}{\multirow{-2}{*}{\citepalias{ellam2018}}} &
    $10^2$ &
    \multirow{-2}{*}{$\boldsymbol{\textcolor[HTML]{FFC20A}{\Lambda}}\vert\mathcal{C},\mathbf{y}$} &
    - &
    - &
    $0.70$ &
    $0.70$ &
    $0.40$ \\ \Xcline{2-10}{0.01pt}
   &
     &
    \multicolumn{1}{l}{} &
    $10^4$ &
     &
    - &
    - &
    $0.70$ &
    $\mathbf{0.67}$ &
    $0.77$ \\
   &
     &
    \multicolumn{1}{l}{} &
    $10^2$ &
     &
    - &
    - &
    $0.65$ &
    $\mathbf{0.68}$ &
    $0.85$ \\
  \multirow{-13}{*}{Totally} &
    \multirow{-3}{*}{$\textcolor[HTML]{FFC20A}{\Lambda}_{++}$} &
    \multicolumn{1}{l}{\multirow{-3}{*}{\citepalias{gaskin2023}}} &
    learned &
    \multirow{-3}{*}{$\boldsymbol{\textcolor[HTML]{FFC20A}{\Lambda}}\vert\mathcal{C},\mathbf{y}$} &
    - &
    - &
    $0.64$ &
    $0.70$ &
    $0.84$ \\ \hlineB{3.0}
   &
     &
     &
     &
    $\boldsymbol{\textcolor[HTML]{FFC20A}{\Lambda}}\vert\mathcal{C},\mathbf{y}$ &
    - &
    - &
    $0.72$ &
    $0.74$ &
    $0.20$ \\
   &
     &
     &
    \multirow{-2}{*}{$10^4$} &
    $\mathbf{\textcolor[HTML]{1E88E5}{T}}\vert\mathcal{C},\mathbf{y}$ &
    $0$ &
    $\mathbf{6806}$ &
    $0.72$ &
    $0.69$ &
    $0.67$ \\
   &
     &
     &
     &
    $\boldsymbol{\textcolor[HTML]{FFC20A}{\Lambda}}\vert\mathcal{C},\mathbf{y}$ &
    - &
    - &
    $0.70$ &
    $0.71$ &
    $0.39$ \\
   &
    \multirow{-4}{*}{$\mathbf{\textcolor[HTML]{1E88E5}{T}}_{\cdot +},\textcolor[HTML]{FFC20A}{\Lambda}_{++}$} &
     &
    \multirow{-2}{*}{$10^2$} &
    $\mathbf{\textcolor[HTML]{1E88E5}{T}}\vert\mathcal{C},\mathbf{y}$ &
    $0$ &
    $7067$ &
    $0.71$ &
    $0.63$ &
    $0.69$ \\
   &
     &
     &
     &
    $\boldsymbol{\textcolor[HTML]{FFC20A}{\Lambda}}\vert\mathcal{C},\mathbf{y}$ &
    - &
    - &
    $\mathbf{0.74}$ &
    $0.69$ &
    $0.23$ \\
   &
     &
     &
    \multirow{-2}{*}{$10^4$} &
    $\mathbf{\textcolor[HTML]{1E88E5}{T}}\vert\mathcal{C},\mathbf{y}$ &
    $0$ &
    $6819$ &
    $0.72$ &
    $0.69$ &
    $0.68$ \\
   &
     &
     &
     &
    $\boldsymbol{\textcolor[HTML]{FFC20A}{\Lambda}}\vert\mathcal{C},\mathbf{y}$ &
    - &
    - &
    $0.72$ &
    $\mathbf{0.62}$ &
    $0.42$ \\
   &
    \multirow{-4}{*}{$\mathbf{\textcolor[HTML]{1E88E5}{T}}_{\cdot +},\textcolor[HTML]{FFC20A}{\Lambda}_{++},\boldsymbol{\textcolor[HTML]{FFC20A}{\Lambda}}_{+\cdot}$} &
    \multirow{-8}{*}{This work} &
    \multirow{-2}{*}{$10^2$} &
    $\mathbf{\textcolor[HTML]{1E88E5}{T}}\vert\mathcal{C},\mathbf{y}$ &
    $0$ &
    $6944$ &
    $0.71$ &
    $\mathbf{0.62}$ &
    $0.71$ \\ \Xcline{2-10}{0.01pt}
   &
     &
    \multicolumn{1}{l}{} &
    $10^4$ &
     &
    - &
    - &
    $\mathbf{0.74}$ &
    $0.69$ &
    $0.22$ \\
   &
    \multirow{-2}{*}{$\boldsymbol{\textcolor[HTML]{FFC20A}{\Lambda}}_{+\cdot}$} &
    \multicolumn{1}{l}{\multirow{-2}{*}{\citepalias{ellam2018}}} &
    $10^2$ &
    \multirow{-2}{*}{$\boldsymbol{\textcolor[HTML]{FFC20A}{\Lambda}}\vert\mathcal{C},\mathbf{y}$} &
    - &
    - &
    $0.72$ &
    $\mathbf{0.62}$ &
    $0.41$ \\ \Xcline{2-10}{0.01pt}
   &
     &
    \multicolumn{1}{l}{} &
    $10^4$ &
     &
    - &
    - &
    $0.71$ &
    $\mathbf{0.61}$ &
    $0.79$ \\
   &
     &
    \multicolumn{1}{l}{} &
    $10^2$ &
     &
    - &
    - &
    $0.66$ &
    $\mathbf{0.62}$ &
    $\mathbf{0.87}$ \\
  \multirow{-13}{*}{Singly} &
    \multirow{-3}{*}{$\boldsymbol{\textcolor[HTML]{FFC20A}{\Lambda}}_{+\cdot}$} &
    \multicolumn{1}{l}{\multirow{-3}{*}{\citepalias{gaskin2023}}} &
    learned &
    \multirow{-3}{*}{$\boldsymbol{\textcolor[HTML]{FFC20A}{\Lambda}}\vert\mathcal{C},\mathbf{y}$} &
    - &
    - &
    $0.65$ &
    $0.65$ &
    $\mathbf{0.86}$ \\ \hlineB{3.0}
   &
     &
     &
     &
    $\boldsymbol{\textcolor[HTML]{FFC20A}{\Lambda}}\vert\mathcal{C},\mathbf{y}$ &
    - &
    - &
    $0.73$ &
    $0.71$ &
    $0.19$ \\
   &
     &
     &
    \multirow{-2}{*}{$10^4$} &
    $\mathbf{\textcolor[HTML]{1E88E5}{T}}\vert\mathcal{C},\mathbf{y}$ &
    $182988$ &
    $\mathbf{5912}$ &
    $\mathbf{0.76}$ &
    $\mathbf{0.59}$ &
    $0.68$ \\
   &
     &
     &
     &
    $\boldsymbol{\textcolor[HTML]{FFC20A}{\Lambda}}\vert\mathcal{C},\mathbf{y}$ &
    - &
    - &
    $0.61$ &
    $1.12$ &
    $0.17$ \\
  \multirow{-4}{*}{Doubly} &
    \multirow{-4}{*}{$\mathbf{\textcolor[HTML]{1E88E5}{T}}_{\cdot +},\mathbf{\textcolor[HTML]{1E88E5}{T}}_{+ \cdot},\textcolor[HTML]{FFC20A}{\Lambda}_{++}$} &
    \multirow{-4}{*}{This work} &
    \multirow{-2}{*}{$10^2$} &
    $\mathbf{\textcolor[HTML]{1E88E5}{T}}\vert\mathcal{C},\mathbf{y}$ &
    $182988$ &
    $6370$ &
    $0.73$ &
    $\mathbf{0.59}$ &
    $\mathbf{0.80}$ \\ \hlineB{3.0}
   &
     &
     &
     &
    $\boldsymbol{\textcolor[HTML]{FFC20A}{\Lambda}}\vert\mathcal{C},\mathbf{y}$ &
    - &
    - &
    $0.72$ &
    $0.71$ &
    $0.19$ \\
   &
     &
     &
    \multirow{-2}{*}{$10^4$} &
    $\mathbf{\textcolor[HTML]{1E88E5}{T}}\vert\mathcal{C},\mathbf{y}$ &
    $119421$ &
    $\mathbf{5314}$ &
    $\mathbf{0.78}$ &
    $\mathbf{0.55}$ &
    $\mathbf{0.79}$ \\
   &
     &
     &
     &
    $\boldsymbol{\textcolor[HTML]{FFC20A}{\Lambda}}\vert\mathcal{C},\mathbf{y}$ &
    - &
    - &
    $0.63$ &
    $1.10$ &
    $0.13$ \\
  \multirow{-4}{*}{Doubly and 10\% cell} &
    \multirow{-4}{*}{$\mathbf{\textcolor[HTML]{1E88E5}{T}}_{\cdot +},\mathbf{\textcolor[HTML]{1E88E5}{T}}_{+\cdot},\mathbf{\textcolor[HTML]{1E88E5}{T}}_{\mathcal{X}_1},\textcolor[HTML]{FFC20A}{\Lambda}_{++}$} &
    \multirow{-4}{*}{This work} &
    \multirow{-2}{*}{$10^2$} &
    $\mathbf{\textcolor[HTML]{1E88E5}{T}}\vert\mathcal{C},\mathbf{y}$ &
    $119421$ &
    $5449$ &
    $\mathbf{0.77}$ &
    $\mathbf{0.56}$ &
    $\mathbf{0.78}$ \\ \hlineB{3.0}
   &
     &
     &
     &
    $\boldsymbol{\textcolor[HTML]{FFC20A}{\Lambda}}\vert\mathcal{C},\mathbf{y}$ &
    - &
    - &
    $0.72$ &
    $0.71$ &
    $0.19$ \\
   &
     &
     &
    \multirow{-2}{*}{$10^4$} &
    $\mathbf{\textcolor[HTML]{1E88E5}{T}}\vert\mathcal{C},\mathbf{y}$ &
    $74314$ &
    $\mathbf{4521}^{\mbox{\large $*$}}$ &
    $\mathbf{0.81}^{\mbox{\large $*$}}$ &
    $\mathbf{0.51}^{\mbox{\large $*$}}$ &
    $0.77$ \\
   &
     &
     &
     &
    $\boldsymbol{\textcolor[HTML]{FFC20A}{\Lambda}}\vert\mathcal{C},\mathbf{y}$ &
    - &
    - &
    $0.64$ &
    $1.05$ &
    $0.21$ \\
  \multirow{-4}{*}{Doubly and 20\% cell} &
    \multirow{-4}{*}{$\mathbf{\textcolor[HTML]{1E88E5}{T}}_{\cdot +},\mathbf{\textcolor[HTML]{1E88E5}{T}}_{+\cdot},\mathbf{\textcolor[HTML]{1E88E5}{T}}_{\mathcal{X}_2},\textcolor[HTML]{FFC20A}{\Lambda}_{++}$} &
    \multirow{-4}{*}{This work} &
    \multirow{-2}{*}{$10^2$} &
    $\mathbf{\textcolor[HTML]{1E88E5}{T}}\vert\mathcal{C},\mathbf{y}$ &
    $74314$ &
    $4543$ &
    $\mathbf{0.80}$ &
    $\mathbf{0.51}^{\mbox{\large $*$}}$ &
    $\mathbf{0.81}$ \\ \hlineB{3.0}

%% file: tex_figures/log_destination_attraction_predictions/table_69x13_JointTableSIMLatentMCMC_log_destination_attraction_predictions_burnin_10000_thinning_1_N_100000_low_noise.tex
\begin{tikzpicture}

\pgfplotsset{
    table/search path={./tex_tables/log_destination_attraction_predictions/},
}
\pgfdeclareplotmark{mystar}{
    \node[star,star point ratio=2.25,minimum size=5pt,inner sep=0pt,opacity=0.75,draw=cpurple,solid,fill=cpurple] {};
  }

\begin{axis}[
legend style={
  fill opacity=0.8,
  draw opacity=1,
  text opacity=1,
  at={(0.03,0.97)},
  anchor=north west,
  draw=none
},
tick pos=both,
x grid style={darkgray176},
xlabel={\(\displaystyle \mathbb{{E}}[\mathbf{x}\vert\mathbf{y}]\)},
xmin=-4.5, xmax=-0.5,
xtick={-4.5,-4,-3.5,-3,-2.5,-2,-1.5,-1,-0.5},
y grid style={darkgray176},
ylabel={\(\displaystyle \log{\mathbf{y}}\)},
ymin=-4.5, ymax=-0.5,
ytick={-4.5,-4,-3.5,-3,-2.5,-2,-1.5,-1,-0.5},
]


\addplot [fill=cskyblue, draw=cskyblue, opacity=1.0, fill opacity=0.0, mark=*, only marks]
table {type_JointTableSIMLatentMCMC_noise_regime_low_experiment_title_grand_total.dat};
\label{subplot:type_JointTableSIMLatentMCMC_noise_regime_low_experiment_title_grand_total}


\addplot [fill=cdarkblue, draw=cdarkblue, opacity=1.0, fill opacity=0.0, mark=triangle*, only marks]
table {type_JointTableSIMLatentMCMC_noise_regime_low_experiment_title_row_margin.dat};
\label{subplot:type_JointTableSIMLatentMCMC_noise_regime_low_experiment_title_row_margin}


\addplot [fill=cgrassgreen, draw=cgrassgreen, opacity=1.0, fill opacity=0.0, mark=square*, only marks]
table {type_JointTableSIMLatentMCMC_noise_regime_low_experiment_title_both_margins.dat};
\label{subplot:type_JointTableSIMLatentMCMC_noise_regime_low_experiment_title_both_margins}


\addplot [fill=corange, draw=corange, opacity=1.0, fill opacity=0.0, mark=pentagon*, only marks]
table {type_JointTableSIMLatentMCMC_noise_regime_low_experiment_title_both_margins_permuted_cells_10percent.dat};
\label{subplot:type_JointTableSIMLatentMCMC_noise_regime_low_experiment_title_both_margins_permuted_cells_10percent}


\addplot [fill=cpurple, draw=cpurple, opacity=1.0, fill opacity=0.0, mark=diamond*, only marks]
table {type_JointTableSIMLatentMCMC_noise_regime_low_experiment_title_both_margins_permuted_cells_20percent.dat};
\label{subplot:type_JointTableSIMLatentMCMC_noise_regime_low_experiment_title_both_margins_permuted_cells_20percent}

\draw[variable=\x]  plot ({\x}, {\x});

\end{axis}

\matrix [
            draw,
            matrix of nodes,
            column sep=0.0,
            fill opacity=0.0,
            draw opacity=0.0,
            text opacity=1.0,
            anchor=north west,
            node font=\small,
            column 1/.style={anchor=west}
            ] at (0.1,5.6) {
                \ref*{subplot:type_JointTableSIMLatentMCMC_noise_regime_low_experiment_title_both_margins_permuted_cells_20percent}
                $\mathcal{C}_T=\left\{\mathbf{T}_{\cdot+}\;,\mathbf{T}_{+,\cdot}\;,\mathbf{T}_{\mathcal{X}_2}\right\}\; R^2 = 0.84$ \\
                \ref*{subplot:type_JointTableSIMLatentMCMC_noise_regime_low_experiment_title_both_margins_permuted_cells_10percent} $\mathcal{C}_T=\left\{\mathbf{T}_{\cdot+}\;,\mathbf{T}_{+,\cdot}\;,\mathbf{T}_{\mathcal{X}_1}\right\} \; R^2 = 0.84$ \\
                \ref*{subplot:type_JointTableSIMLatentMCMC_noise_regime_low_experiment_title_both_margins} $\mathcal{C}_T=\left\{\mathbf{T}_{\cdot+}\;,\mathbf{T}_{+,\cdot}\right\} \; R^2 = 0.83$ \\
                \ref*{subplot:type_JointTableSIMLatentMCMC_noise_regime_low_experiment_title_row_margin} $\mathcal{C}_T=\left\{\mathbf{T}_{\cdot+}\right\}\; R^2 = 0.77$ \\
                \ref*{subplot:type_JointTableSIMLatentMCMC_noise_regime_low_experiment_title_grand_total} $\mathcal{C}_T=\left\{\mathbf{T}_{++}\right\}\; R^2 = 0.77$ \\
            };
                
\end{tikzpicture}

%% file: tex_figures/log_destination_attraction_predictions/table_69x13_JointTableSIMLatentMCMC_log_destination_attraction_predictions_burnin_10000_thinning_1_N_100000_high_noise.tex
\begin{tikzpicture}

  \pgfplotsset{
      table/search path={./tex_tables/log_destination_attraction_predictions/},
  }
  
  \begin{axis}[
  legend style={
    fill opacity=0.8,
    draw opacity=1,
    text opacity=1,
    at={(0.03,0.97)},
    anchor=north west,
    draw=none
  },
  tick pos=both,
  x grid style={darkgray176},
  xlabel={\(\displaystyle \mathbb{{E}}[\mathbf{x}\vert\mathbf{y}]\)},
  xmin=-4.5, xmax=-0.5,
  xtick={-4.5,-4,-3.5,-3,-2.5,-2,-1.5,-1,-0.5},
  y grid style={darkgray176},
  ylabel={\(\displaystyle \log{\mathbf{y}}\)},
  ymin=-4.5, ymax=-0.5,
  ytick={-4.5,-4,-3.5,-3,-2.5,-2,-1.5,-1,-0.5},
  ]

  
  \addplot [draw=cskyblue, fill=cskyblue, mark=*, opacity=1.0, fill opacity=0.0, only marks]
  table {type_JointTableSIMLatentMCMC_noise_regime_high_experiment_title_grand_total.dat};
  \label{subplot:type_JointTableSIMLatentMCMC_noise_regime_high_experiment_title_grand_total}
  
  
  \addplot [draw=cdarkblue, fill=cdarkblue, mark=triangle, opacity=1.0, fill opacity=0.0, only marks]
  table {type_JointTableSIMLatentMCMC_noise_regime_high_experiment_title_row_margin.dat};
  \label{subplot:type_JointTableSIMLatentMCMC_noise_regime_high_experiment_title_row_margin}
  
  
  \addplot [draw=cgrassgreen, fill=cgrassgreen, mark=square, opacity=1.0, fill opacity=0.0, only marks]
  table {type_JointTableSIMLatentMCMC_noise_regime_high_experiment_title_both_margins.dat};
  \label{subplot:type_JointTableSIMLatentMCMC_noise_regime_high_experiment_title_both_margins}
  
  
  \addplot [draw=corange, fill=corange, mark=pentagon, opacity=1.0, fill opacity=0.0, only marks]
  table {type_JointTableSIMLatentMCMC_noise_regime_high_experiment_title_both_margins_permuted_cells_10percent.dat};
  \label{subplot:type_JointTableSIMLatentMCMC_noise_regime_high_experiment_title_both_margins_permuted_cells_10percent}
  
  
  \addplot [draw=cpurple, fill=cpurple, mark=diamond, opacity=1.0, fill opacity=0.0, only marks]
  table {type_JointTableSIMLatentMCMC_noise_regime_high_experiment_title_both_margins_permuted_cells_20percent.dat};
  \label{subplot:type_JointTableSIMLatentMCMC_noise_regime_high_experiment_title_both_margins_permuted_cells_20percent}
  
  \draw[variable=\x]  plot ({\x}, {\x});

  \end{axis}

  \matrix [
            draw,
            matrix of nodes,
            column sep=0.0,
            fill opacity=0.0,
            draw opacity=0.0,
            text opacity=1.0,
            anchor=north west,
            node font=\small,
            column 1/.style={anchor=west}
        ] at (0.1,5.6) {
          \ref*{subplot:type_JointTableSIMLatentMCMC_noise_regime_high_experiment_title_both_margins_permuted_cells_20percent}
          $\mathcal{C}_T=\left\{\mathbf{T}_{\cdot+}\;,\mathbf{T}_{+,\cdot}\;,\mathbf{T}_{\mathcal{X}_2}\right\}\; R^2 = 0.99$ \\
          
          \ref*{subplot:type_JointTableSIMLatentMCMC_noise_regime_high_experiment_title_both_margins_permuted_cells_10percent}
          $\mathcal{C}_T=\left\{\mathbf{T}_{\cdot+}\;,\mathbf{T}_{+,\cdot}\;,\mathbf{T}_{\mathcal{X}_1}\right\}\; R^2 = 0.99$ \\
          
          \ref*{subplot:type_JointTableSIMLatentMCMC_noise_regime_high_experiment_title_both_margins} $\mathcal{C}_T=\left\{\mathbf{T}_{\cdot+},\mathbf{T}_{+\cdot}\right\}\; R^2 = 0.99$ \\
          \ref*{subplot:type_JointTableSIMLatentMCMC_noise_regime_high_experiment_title_row_margin} $\mathcal{C}_T=\left\{\mathbf{T}_{\cdot+}\right\}\; R^2 = 0.99$ \\
          \ref*{subplot:type_JointTableSIMLatentMCMC_noise_regime_high_experiment_title_grand_total} $\mathcal{C}_T=\left\{\mathbf{T}_{++}\right\}\; R^2 = 0.99$ \\
        };
  
  \end{tikzpicture}
  

%% file: tex_figures/mean_table_heatmaps/trip_colorbar.tikz
\begin{tikzpicture}
    \pgfplotsset{
        scaled ticks = false,
        every tick label/.append style={font=\small},
    }

    \begin{axis}[
        axis x line=bottom,
        axis y line=left,
        hide y axis,
        hide x axis,
        tick pos=both,
        scale only axis,
        xmin=0,
        ymin=0.0,
        width=0.7\textwidth,
        height=\axisdefaultheight,
        xtick = \empty,
        point meta min=0.0,
        point meta max=1.0,
        colormap name=yellowpurple,
        colorbar horizontal,
        colorbar style={
            title={$\textcolor[HTML]{1E88E5}{\mathbf{T}}^{\mathcal{D}}$},
            title style={at={(-0.44,-0.1)},anchor=west},
            xmin=0,
            xmax=1,
            ymin=0,
            height=0.35cm,
            ytick = \empty,
            xtick= {0.000000,0.125000,0.250000,0.375000,0.500000,0.537538,0.575075,0.612613,0.650150,0.687688,
            0.725225,0.762763,0.800300,0.837838,0.875375,0.912913,0.950450,0.987988},
            xticklabels={0,25,50,75,100,125,150,175,200,225,250,275,300,325,350,375,400,425},
        }
        ]
        \addplot[draw=none,fill=none] coordinates {(0,0)};

    \end{axis}

\end{tikzpicture}

%% file: tex_figures/mean_table_heatmaps/destination_demand_colorbar.tikz
\begin{tikzpicture}
        \pgfplotsset{
            scaled ticks = false,
            every tick label/.append style={font=\small},
        }
        \begin{axis}[
                axis x line=top,
                axis y line=right,
                hide y axis,
                hide x axis,
                tick pos=both,
                scale only axis,
                xmin=0,
                ymin=0.0,
                width=0.7\textwidth,
                height=\axisdefaultheight,
                xtick = \empty,
                point meta min=0.0,
                point meta max=1.0,
                colormap name=bluegreen,
                colorbar horizontal,
                colorbar style={
                    title={$\textcolor[HTML]{1E88E5}{\mathbf{T}}_{\cdot +}^{\mathcal{D}}\cdot 10^2$},
                    title style={at={(-0.44,-0.1)},anchor=west},
                    xmin=0,
                    xmax=1,
                    ymin=0,
                    height=0.35cm,
                    ytick=\empty,
                    xtick={0.021227,0.074966,0.128706,0.182445,0.236185,0.289924,0.343663,0.397403,0.451142,0.504881,
                    0.558621,0.612360,0.666099,0.719839,0.773578,0.827318,0.881057,0.934796,0.988536},
                    xticklabels={6,12,18,24,30,36,42,48,54,60,66,72,78,84,90,96,102,108,114},
                }
            ]
            \addplot[draw=none,fill=none] coordinates {(0,0)};

        \end{axis}
        
\end{tikzpicture}

%% file: tex_figures/mean_table_heatmaps/trip_error_colorbar.tikz
\begin{tikzpicture}
    \pgfplotsset{
        scaled ticks = false,
        every tick label/.append style={font=\small},
    }
    \pgfplotstableread{./tex_tables/mean_table_heatmaps/trip_error_colorbar_ticks.txt}{\colorbarticks}
    \pgfplotstableset{
        sci,
        sci zerofill,
        sci sep align,
        precision=1,
        sci superscript,
    }

    \begin{axis}[
        axis x line=bottom,
        axis y line=right,
        hide y axis,
        hide x axis,
        tick pos=both,
        scale only axis,
        xmin=0,
        ymin=0.0,
        width=0.7\textwidth,
        height=\axisdefaultheight,
        xtick = \empty,
        point meta min=0.0,
        point meta max=1.0,
        colormap name=redwhiteblue,
        colorbar horizontal,
        colorbar style={
            title={$\left(\textcolor[HTML]{1E88E5}{\mathbf{T}}^{\mathcal{D}}-\mathbb{E}\left[ h\left(\textcolor[HTML]{1E88E5}{\mathbf{T}}^{(1\mathrel{\mathop\ordinarycolon}N)},\textcolor[HTML]{FFC20A}{\boldsymbol{\Lambda}}^{(1\mathrel{\mathop\ordinarycolon}N)}\right) \right]\right) \cdot 10^{-3}$},
            title style={at={(-0.44,-0.1)},anchor=west},
            xmin=0,
            xmax=1,
            ymin=0,
            height=0.35cm,
            ytick=\empty,
            xtick={0.000000,0.012383,0.026423,0.042630,0.061799,0.085260,0.115506,0.158136,0.231012,0.500000,0.768988,0.841864,0.884494,0.914740,0.938201,0.957370,0.973577,0.987617,1.000000},
            xticklabels={-9,-8,-7,-6,-5,-4,-3,-2,-1,0,1,2,3,4,5,6,7,8,9},
        }
    ]
    \addplot[draw=none,fill=none] coordinates {(0,0)};

    \end{axis}

\end{tikzpicture}

%% file: tex_figures/mean_table_heatmaps/destination_demand_error_colorbar.tikz
\begin{tikzpicture}
    \pgfplotsset{
        scaled ticks = false,
        every tick label/.append style={font=\small},
    }

    \begin{axis}[
        axis x line=top,
        axis y line=right,
        hide y axis,
        hide x axis,
        tick pos=both,
        scale only axis,
        xmin=0,
        ymin=0.0,
        width=0.7\textwidth,
        height=\axisdefaultheight,
        xtick = \empty,
        point meta min=0.0,
        point meta max=1.0,
        colormap name=yellowwhitepurple,
        colorbar horizontal,
        colorbar style={
            title={$\left(\textcolor[HTML]{1E88E5}{\mathbf{T}}_{+\cdot}^{\mathcal{D}}-\mathbb{E}\left[ \textcolor[HTML]{FFC20A}{\boldsymbol{\Lambda}}_{+\cdot}^{(1\mathrel{\mathop\ordinarycolon}N)}\right]\right) \cdot 10^{-2}$},
            title style={at={(-0.44,-0.1)},anchor=west},
            xmin=0,
            xmax=1,
            ymin=0,
            height=0.35cm,
            xtick={0.0000004,0.0172434,0.0394744,0.0708064,0.1243684,0.5000004,0.8756324,0.9291944,0.9605264,0.9827574,1.0000004},
            xticklabels={-5,-4,-3,-2,-1,0,1,2,3,4,5},
            ytick=\empty,
            xtick distance={0.01}
        }
    ]
    \addplot[draw=none] coordinates {(0,0)};

    \end{axis}
        
\end{tikzpicture}

%% file: tex_figures/mean_table_heatmaps/ground_truth_table.tikz
\begin{tikzpicture}

    \pgfplotsset{
        compat=newest,
        every tick label/.append style={font=\tiny},
        every non boxed x axis/.append style={x axis line style=-},
        every non boxed y axis/.append style={y axis line style=-},
    }
    \pgfplotstableread{./tex_tables/mean_table_heatmaps/ground_truth_table_cell_data.txt}{\groundtruthtable}
    \pgfplotstableread{./tex_tables/mean_table_heatmaps/ground_truth_table_destination_demand_cell_data.txt}{\groundtruthdestinationdemand}

    \begin{axis}[
                axis x line=top,
                axis y line=right,
                scale only axis,
                y dir=reverse,
                xmax=70,
                width=\textwidth,
                height=4cm,
                xtick = {1,2,3,...,70},
                xtick pos=right,
                xtick align = center,
                xticklabel={
                \ifdim \tick cm < 70 cm
                        $\pgfmathprintnumber{\tick}$
                \fi},
                xticklabel style = {xshift=-0.1cm},
                xtick distance = {1},
                ytick align = center,
                ytick = {1,2,...,14},
                ytick pos=right,
                yticklabel={
                \ifdim \tick cm > 0 cm
                        $\pgfmathprintnumber{\tick}$
                \fi},
                yticklabel style = {yshift=0.13cm},
                ytick distance = {1},
                view={0}{90},
                grid=both,
            ]

            \addplot[
                matrix plot*,
                domain=0:69,
                y domain=0:13,
                mesh/rows=70,
                mesh/cols=14,
                mesh/color input=colormap,
                colormap name=yellowpurple,
                point meta=explicit,
                shift={(-0.5,-0.5)},
            ] table [x=x,y=y,meta=color] {\groundtruthtable};

            \addplot[
                matrix plot*,
                domain=69:70,
                y domain=0:13,
                mesh/rows=2,
                mesh/cols=13 ,
                mesh/color input=colormap,
                colormap name=bluegreen,
                point meta=explicit,
                shift={(0.5,0.5)},
            ] table [x=x,y=y,meta=color] {\groundtruthdestinationdemand};

    \end{axis}
    
    \end{tikzpicture}

%% file: tex_figures/mean_table_heatmaps/table_error_mean_heatmap_JointTableSIMLatentMCMC_high_noise_both_margins_permuted_cells_20percent_thinning1_burnin10000.tikz
\begin{tikzpicture}

    \pgfplotsset{
        compat=newest,
        every tick label/.append style={font=\tiny},
        every non boxed x axis/.append style={x axis line style=-},
        every non boxed y axis/.append style={y axis line style=-},
    }

\pgfplotstableread{./tex_tables/mean_table_heatmaps/table_mean_heatmap_JointTableSIMLatentMCMC_high_noise_both_margins_permuted_cells_20percent_thinning1_burnin10000_cell_data.txt}{\tabledata}
\pgfplotstableread{./tex_tables/mean_table_heatmaps/table_mean_heatmap_JointTableSIMLatentMCMC_high_noise_both_margins_permuted_cells_20percent_thinning1_burnin10000_covered_cell_coordinates.dat}{\tablecoverage}
\pgfplotstableread{./tex_tables/mean_table_heatmaps/table_mean_heatmap_JointTableSIMLatentMCMC_high_noise_both_margins_permuted_cells_20percent_thinning1_burnin10000_fixed_cell_coordinates.dat}{\tablecellsfixed}
    
    \begin{axis}[
            axis x line=top,
            axis y line=right,
            scale only axis,
            y dir=reverse,
            xmin=0,
            ymin=0.0,
            xmax=70,
            width=\textwidth,
            height=4cm,
            xtick = {1,2,3,...,70},
            xtick style={color=black},
            xtick align = center,
            xticklabel={
            \ifdim \tick cm < 70 cm
                   $\pgfmathprintnumber{\tick}$
            \fi},
            xticklabel style = {xshift=-0.1cm},
            xtick distance = {1},
            ytick align = center,
            ytick = {1,2,...,14},
            yticklabel={
            \ifdim \tick cm > 0 cm
                    $\pgfmathprintnumber{\tick}$
            \fi},
            yticklabel style = {yshift=0.13cm},
            ytick distance = {1},
            view={0}{90},
            grid=both,
            point meta min=0.0,
            point meta max=1.0,
            colormap name=redwhiteblue,
            enlargelimits=false,
        ]
        
        \addplot[
            matrix plot*,
            domain=0:69,
            y domain=0:13,
            mesh/rows=70,
            mesh/cols=14,
            mesh/color input=colormap,
            point meta=explicit,
            shift={(-0.5,-0.5)},
        ] table [x=x,y=y,meta=color] {\tabledata};
                
        \pgfdeclareplotmark{checkmark}
        {
        \node[scale=0.8] at (0,0.0) {\checkmark};
        }
        \addplot [
        only marks,
        mark=checkmark,
        xshift=0.13cm,
        yshift=-0.15cm,
        ] table [x=x, y=y] {\tablecoverage};

        \addplot [
            only marks,
            mark=square*,
            mark options={
                fill opacity=0.0,
                draw=black,
                line width=0.2pt,
                xscale=1.81,
                yscale=2.1,
            },
            xshift=0.1275cm,
            yshift=-0.152cm,
            ] table [x=x, y=y] {\tablecellsfixed};
            
    \end{axis}
    
    \end{tikzpicture}

%% file: tex_figures/mean_table_heatmaps/intensity_error_mean_heatmap_NeuralABM_low_noise_row_margin_thinning1_burnin10000.tikz
\begin{tikzpicture}

    \pgfplotsset{
        compat=newest,
        every tick label/.append style={font=\tiny},
        every non boxed x axis/.append style={x axis line style=-},
        every non boxed y axis/.append style={y axis line style=-},
    }
    \pgfplotstableread{./tex_tables/mean_table_heatmaps/intensity_mean_heatmap_NeuralABM_low_noise_row_margin_thinning1_burnin10000_cell_data.txt}{\neuralabmintensitydata}
\pgfplotstableread{./tex_tables/mean_table_heatmaps/intensity_mean_heatmap_NeuralABM_low_noise_row_margin_thinning1_burnin10000_covered_cell_coordinates.dat}{\neuralabmcoveragedata}
\pgfplotstableread{./tex_tables/mean_table_heatmaps/intensity_mean_heatmap_NeuralABM_low_noise_row_margin_thinning1_burnin10000_destination_demand_cell_data.txt}{\neuralabmdestinationdemanddata}

    \begin{axis}[
                axis x line=top,
                axis y line=right,
                scale only axis,
                y dir=reverse,
                xmax=70,
                width=\textwidth,
                height=4cm,
                xtick = {1,2,3,...,70},
                xtick pos=right,
                xtick align = center,
                xticklabel={
                \ifdim \tick cm < 70 cm
                        $\pgfmathprintnumber{\tick}$
                \fi},
                xticklabel style = {xshift=-0.1cm},
                xtick distance = {1},
                ytick align = center,
                ytick = {1,2,...,14},
                ytick pos=right,
                yticklabel={
                \ifdim \tick cm > 0 cm
                        $\pgfmathprintnumber{\tick}$
                \fi},
                yticklabel style = {yshift=0.13cm},
                ytick distance = {1},
                view={0}{90},
                grid=both,
            ]

            \addplot[
                matrix plot*,
                domain=0:69,
                y domain=0:13,
                mesh/rows=70,
                mesh/cols=14,
                mesh/color input=colormap,
                colormap name=redwhiteblue,
                point meta=explicit,
                shift={(-0.5,-0.5)},
            ] table [x=x,y=y,meta=color] {\neuralabmintensitydata};

            \addplot[
                matrix plot*,
                domain=69:70,
                y domain=0:13,
                mesh/rows=2,
                mesh/cols=13,
                mesh/color input=colormap,
                colormap name=yellowwhitepurple,
                point meta=explicit,
                shift={(0.5,0.5)},
            ] table [x=x,y=y,meta=color] {\neuralabmdestinationdemanddata};

        \pgfdeclareplotmark{checkmark}
        {
        \node[scale=0.8] at (0,0.0) {\checkmark};
        }
        \addplot [
        only marks,
        mark=checkmark,
        xshift=0.13cm,
        yshift=-0.15cm,
        ] table [x=x, y=y] {\neuralabmcoveragedata};

    \end{axis}
    
    \end{tikzpicture}

%% file: tex_figures/mean_table_heatmaps/intensity_error_mean_heatmap_SIMLatentMCMC_high_noise_row_margin_thinning1_burnin10000.tikz
\begin{tikzpicture}

        \pgfplotsset{
            compat=newest,
            every tick label/.append style={font=\tiny},
            every non boxed x axis/.append style={x axis line style=-},
            every non boxed y axis/.append style={y axis line style=-},
        }
        \pgfplotstableread{./tex_tables/mean_table_heatmaps/intensity_mean_heatmap_SIMLatentMCMC_high_noise_row_margin_thinning1_burnin10000_cell_data.txt}{\simintensitydata}
        \pgfplotstableread{./tex_tables/mean_table_heatmaps/intensity_mean_heatmap_SIMLatentMCMC_high_noise_row_margin_thinning1_burnin10000_covered_cell_coordinates.dat}{\simcoveragedata}
        \pgfplotstableread{./tex_tables/mean_table_heatmaps/intensity_mean_heatmap_SIMLatentMCMC_high_noise_row_margin_thinning1_burnin10000_destination_demand_cell_data.txt}{\simdestinationdemanddata}
    
        \begin{axis}[
                    axis x line=top,
                    axis y line=right,
                    scale only axis,
                    y dir=reverse,
                    xmax=70,
                    width=\textwidth,
                    height=4cm,
                    xtick = {1,2,3,...,70},
                    xtick pos=right,
                    xtick align = center,
                    xticklabel={
                    \ifdim \tick cm < 70 cm
                            $\pgfmathprintnumber{\tick}$
                    \fi},
                    xticklabel style = {xshift=-0.1cm},
                    xtick distance = {1},
                    ytick align = center,
                    ytick = {1,2,...,14},
                    ytick pos=right,
                    yticklabel={
                    \ifdim \tick cm > 0 cm
                            $\pgfmathprintnumber{\tick}$
                    \fi},
                    yticklabel style = {yshift=0.13cm},
                    ytick distance = {1},
                    view={0}{90},
                    grid=both,
                ] 
                
                \addplot[
                matrix plot*,
                domain=0:69,
                y domain=0:13,
                mesh/rows=70,
                mesh/cols=14,
                mesh/color input=colormap,
                colormap name=redwhiteblue,
                point meta=explicit,
                shift={(-0.5,-0.5)},
            ] table [x=x,y=y,meta=color] {\simintensitydata};

            \addplot[
                matrix plot*,
                domain=69:70,
                y domain=0:13,
                mesh/rows=2,
                mesh/cols=13,
                mesh/color input=colormap,
                colormap name=yellowwhitepurple,
                point meta=explicit,
                shift={(0.5,0.5)},
            ] table [x=x,y=y,meta=color] {\simdestinationdemanddata};
    
            \pgfdeclareplotmark{checkmark}
            {
            \node[scale=0.8] at (0,0.0) {\checkmark};
            }
            \addplot [
            only marks,
            mark=checkmark,
            xshift=0.13cm,
            yshift=-0.15cm,
            ] table [x=x, y=y] {\simcoveragedata};

        \end{axis}
        
        \end{tikzpicture}

%% file: tex_figures/intensity_and_table_convergence/mean_table_and_intensity_relative_l_1_norm_vs_mcmc_iteration.tex
\begin{tikzpicture}

\pgfplotsset{
    table/search path={./tex_tables/intensity_and_table_convergence/},
}

\begin{axis}[
    legend style={fill opacity=0.8, draw opacity=1, text opacity=1, draw=none},
    tick pos=both,
    x grid style={darkgray176},
    xlabel={MCMC Iteration},
    restrict x to domain=0:79920,
    xtick distance={10000},
    y grid style={darkgray176},
    ylabel={$\bigg\| \mathbb{{E}}[\mathbf{T}\vert\mathcal{C},\mathbf{y}] - \mathbb{{E}}[\boldsymbol{\Lambda}\vert\mathcal{C},\mathbf{y}] \bigg\|_1^{2}$},
    restrict y to domain=0.00:0.01,
    ytick distance={0.001}
]






\addplot [limegreen018569,line width=0.38mm] table {experiment_title_both_margins_noise_regime_low.dat};
\addlegendentry{$\mathcal{C}=\{\mathbf{T}_{\cdot+}\;,\mathbf{T}_{+,\cdot}\}$}
\label{subplot:doubly_constrained_low_noise_convergence}

\addplot [limegreen018569,dashed,line width=0.38mm] table {experiment_title_both_margins_noise_regime_high.dat};
\addlegendentry{$\mathcal{C}=\{\mathbf{T}_{\cdot+}\;,\mathbf{T}_{+,\cdot}\}$}
\label{subplot:doubly_constrained_high_noise_convergence}

\addplot [darkorange2551490,line width=0.38mm] table {experiment_title_both_margins_permuted_cells_10percent_noise_regime_low.dat};
\addlegendentry{$\mathcal{C}=\{\mathbf{T}_{\cdot+}\;,\mathbf{T}_{+,\cdot}\;,\mathbf{T}_{\mathcal{X}_1}\}$}
\label{subplot:doubly_10percent_cell_constrained_low_noise_convergence}

\addplot [darkorange2551490,dashed,line width=0.38mm] table {experiment_title_both_margins_permuted_cells_10percent_noise_regime_high.dat};
\addlegendentry{$\mathcal{C}=\{\mathbf{T}_{\cdot+}\;,\mathbf{T}_{+,\cdot}\;,\mathbf{T}_{\mathcal{X}_1}\}$}
\label{subplot:doubly_10percent_cell_constrained_high_noise_convergence}

\addplot [orangered255440,line width=0.38mm] table {experiment_title_both_margins_permuted_cells_20percent_noise_regime_low.dat};
\addlegendentry{$\mathcal{C}=\{\mathbf{T}_{\cdot+}\;,\mathbf{T}_{+,\cdot}\;,\mathbf{T}_{\mathcal{X}_2}\}$}
\label{subplot:doubly_20percent_cell_constrained_low_noise_convergence}

\addplot [orangered255440,dashed,line width=0.38mm] table {experiment_title_both_margins_permuted_cells_20percent_noise_regime_high.dat};
\addlegendentry{$\mathcal{C}=\{\mathbf{T}_{\cdot+}\;,\mathbf{T}_{+,\cdot}\;,\mathbf{T}_{\mathcal{X}_2}\}$}
\label{subplot:doubly_20percent_cell_constrained_high_noise_convergence}

\end{axis}

\end{tikzpicture}